\documentclass[12pt]{article}
\usepackage[right=2cm,left=2cm,top=2cm,bottom=2cm]{geometry}
\usepackage{hyperref}
\hypersetup{colorlinks, citecolor=blue, filecolor=blue, linkcolor=blue, urlcolor=blue}

\usepackage{adjustbox}
\usepackage{caption}
\usepackage{bm}
\usepackage{algorithm}
\usepackage[noend]{algpseudocode}
\usepackage{tcolorbox}
\usepackage[stable]{footmisc}
\usepackage{comment}
\usepackage{threeparttable}

\usepackage{graphicx}
\usepackage{url}
\usepackage[round]{natbib}
\usepackage{amsmath,amssymb,amsfonts,amsthm}
\usepackage{mathtools}
\usepackage{engord}
\usepackage{float}
\usepackage{subfig}
\usepackage{pdflscape}
\usepackage{booktabs}
\usepackage{pgfplots}
\usepackage{amssymb}
\usepackage{multirow}
\usepackage{xr}
\usepackage{xcolor}
\pgfplotsset{compat=1.14}
\pgfplotsset{every axis label/.append style={font=\tiny}}

\theoremstyle{plain}

\newtheorem{lemma}{Lemma}
\newtheorem{proposition}{Proposition}
\newtheorem{corollary}{Corollary}

\theoremstyle{definition}
\newtheorem{remark}{Remark}
\newtheorem{definition}{Definition}
\newtheorem{example}{Example}
\newtheorem{assumption}{Assumption}

\DeclareMathOperator{\vect}{vec}

\usepackage{setspace}
\usepackage{sectsty}
\sectionfont{\large}
\subsectionfont{\normalsize}
\subsubsectionfont{\normalsize}

\usepackage{titlesec}
\titleformat{\section}{\Large\bfseries}{\thesection}{1em}{}

\usepackage{enumitem}
\setdescription{leftmargin=\parindent,labelindent=\parindent}
\setlist[enumerate]{topsep=0pt}

\title{Sensitivity, Informativeness, and Misspecification \\in GMM Estimation}

\author{
Fangzhou Yu\thanks{School of Economics, University of Sydney.
	\href{fangzhou.yu@student.unsw.edu.au}{yfz.1017@gmail.com}} \and Seojeong Lee\thanks{Department of Economics, Seoul National University.
	\href{s.jay.lee@snu.ac.kr}{s.jay.lee@snu.ac.kr}}
}

\date{\today}

\begin{document}

%\bgroup
%\let\footnoterule\relax

%\begin{singlespace}
\maketitle
\begin{abstract}
    This paper develops misspecification-robust sensitivity and informativeness diagnostics for GMM estimators, evaluated at pseudo-true values. The sensitivity matrix nests that of \citet*{andrews2017measuring} under correct specification. The informativeness $\Delta$ measures the share of an estimator's asymptotic variance explained by sampling variation in the moments, a notion of structural efficiency that equals one under correct specification and can fall below one under misspecification, even when the Hansen $J$-test does not reject. We derive influence-function representations for one-step, two-step, iterated, and continuously updating GMM. We show that in minimum-distance estimation, estimating the optimal weight matrix adds estimator variance that the moments do not explain, lowering informativeness, while simpler weight matrices largely avoid it. The choice of weight matrix therefore involves a trade-off between classical efficiency and informativeness. In applications to the automobile demand model of \citet*{berry1995automobile}, the consumption insurance model of \citet*{blundell2008consumption}, and the income-and-democracy regressions of \citet*{acemoglu2008income}, misspecification reorders sensitivity rankings, simpler weights preserve the informativeness that the optimal weight loses, and $\Delta$ detects structural-efficiency losses that the $J$-test does not.
\end{abstract}


\thispagestyle{empty}
%\end{singlespace}

\clearpage
%\egroup
\setcounter{page}{1}

%% Temporary tool to track how this paper is structured. Feel free to comment in or out. 
% \tableofcontents
% \bigskip

%%%%%%%%%%%%%%%%%%%%%%%%%%%%%%%%%%%%%%%%%%%%%%%%%%%%%%%%%%%%%

\section{Introduction}\label{sec:intro}

Assessing sensitivity to model specification is central to evaluating the robustness of empirical conclusions. In reduced-form settings, researchers have access to a rich toolkit: for example, the omitted variable bias formula in linear regression provides a transparent benchmark for assessing the impact of confounding variables, and formal sensitivity diagnostics have been developed for a range of reduced-form estimators \citep*{rosenbaum1983assessing, imbens2003sensitivity, oster2019unobservable, cinelli2020making}. Comparable tools for structural estimators are far less standardized. The relationship between parameters and identifying moment conditions is implicit and often nonlinear, and misspecification is prevalent in applied work.

Empirical researchers increasingly rely on sensitivity diagnostics to assess how GMM estimates depend on particular moment conditions. In many applications, however, overidentifying restrictions are strongly rejected, and the maintained moment conditions fail to hold exactly at any parameter value. In such settings, the GMM estimator does not converge to a structurally ``true'' parameter satisfying the moment restrictions; it converges instead to a pseudo-true value defined as the minimizer of a misspecified population objective function. Local sensitivity can no longer be interpreted as a perturbation around a correctly specified benchmark, and the relevant object is the mapping from the moment conditions to the estimator's misspecified probability limit. This paper develops sensitivity and informativeness measures defined relative to the pseudo-true value and shows that this change of benchmark gives the diagnostics different economic content and interpretation.

A growing literature has begun to formalize sensitivity analysis for structural estimators using large-sample approximations. \citet*{gentzkow2014measuring} introduce the concepts of asymptotic sensitivity and asymptotic sufficiency, which characterize how parameter estimates covary with auxiliary sample statistics. \citet*{andrews2017measuring} (hereafter AGS) formalize sensitivity for GMM through local perturbations of the moment conditions, mapped linearly into asymptotic bias by a closed-form sensitivity matrix. These approaches evaluate sensitivity at a benchmark where the moment conditions are correctly specified. The natural reference point under misspecification is the pseudo-true value, which retains economic meaning in several settings of empirical interest. For instance, in overidentified IV with heterogeneous treatment effects, the 2SLS estimator converges to a weighted average of the local average treatment effects \citep*{angrist1996identification, lee2018consistent}.

Recent work argues that applied researchers should default to misspecification-robust inference for overidentified GMM \citep*{andrews2025purpose}. We take this recommendation as our starting point and ask what additional diagnostic information the misspecification-robust influence function provides. We propose a misspecification-robust sensitivity (MRS) matrix $\Lambda$ that characterizes how local deviations in the moment conditions affect the first-order behavior of GMM estimators around their pseudo-true probability limits. To construct MRS, we derive influence function representations for one-step, two-step, and iterated GMM, explicitly accounting for misspecification and for variation arising from estimated weight matrices; continuously updating GMM is treated in Appendix~\ref{sec:cugmm}. Relative to AGS, MRS contains additional terms operating through the Jacobian of the moment conditions and the estimated weight matrix; these terms are closely related to components of the misspecification-robust asymptotic variance \citep*[e.g.,][]{imbens1997one, hall2003large, hansen2021inference, hwang2022doubly}.

The misspecification-robust influence function also yields an informativeness measure $\Delta_k$, defined as the population $R^2$ from the linear projection of the estimator's influence function on the moments. \citet*{gentzkow2014measuring} show that the analogous quantity, called sufficiency in their framework and informativeness in \citet*{andrews2020informativeness}, equals one under correct specification; this equivalence breaks down under misspecification. We reinterpret $\Delta_k$ as a measure of \textit{structural efficiency}: the proportion of the estimator's asymptotic variance attributable to sampling variation in the moments, as opposed to variation arising from the Jacobian or the estimated weight matrix. Unlike classical GMM efficiency, which concerns the choice of optimal weights to minimize variance, structural efficiency quantifies a different source of variance loss. This loss exists even at the efficient weight matrix when the model is misspecified. $\Delta_k$ is complementary to the Hansen $J$-test \citep*{hansen1982large}. The $J$-test asks whether the moment restrictions are jointly violated; $\Delta_k$ asks how much of the estimator's variance is left unexplained by the moments after such a violation. The income-and-democracy application of \citet*{acemoglu2008income} (AJRY) in Section~\ref{sec:ajry} makes the distinction concrete. At the iterated fixed point the income coefficient has informativeness $\widehat\Delta_\gamma=0.77$: only $77\%$ of its asymptotic variance is explained by the moments, the rest a structural-efficiency loss the $J$-test does not quantify. $\widehat\Delta_\gamma$ is invariant to the centering of the weight matrix, whereas the $J$-test verdict flips with it, from failing to reject without centering ($p=0.42$) to rejecting with centering ($p=0.007$).

We make three contributions. First, we define a projection-based sensitivity measure for GMM estimators evaluated at pseudo-true values that nests the AGS sensitivity matrix under correct specification. Second, we derive misspecification-robust influence functions for one-step, two-step, and iterated GMM, organize them as a channel decomposition into moment, Jacobian, weight-matrix, and first-step components, and use this decomposition to characterize when informativeness falls strictly below one. Third, we apply the framework to three canonical structural settings: the automobile demand model of \citet*{berry1995automobile} (BLP), the consumption insurance model of \citet*{blundell2008consumption} (BPP), and the AJRY income-and-democracy regressions. Accounting for misspecification reorders sensitivity rankings in BLP; in BPP it lowers $\widehat\Delta$ below one for every parameter under the optimal weight but leaves it near one under the diagonally weighted scheme actually used in BPP; and in AJRY it reveals structural-efficiency loss at the iterated fixed point that is invariant to the weight-matrix centering on which the conclusion of the Hansen $J$-test depends. The diagnostics are computable directly from the same inputs as misspecification-robust standard errors.

As a by-product, the framework rationalizes a widespread practice in minimum distance estimation: since the efficient weight matrix contributes a misspecification channel that simpler weight matrices largely avoid (Proposition~\ref{prp:eff_channel}), the common use of diagonally weighted or fixed-weight estimators trades efficiency for informativeness. This trade-off is an asymptotic counterpart to the small-sample bias argument of \citet*{altonji1996small}.

Our paper is also related to a recent literature on local sensitivity in structural models, including \citet*{jorgensen2023sensitivity, iskrev2019expect} on calibrated parameters, \citet*{christensen2023counterfactual} on counterfactual conclusions, \citet*{armstrong2021sensitivity} on moment selection, and \citet*{bonhomme2022minimizing} on estimators designed to minimize the impact of misspecification. Our contribution is complementary in focusing on sensitivity to moment conditions evaluated at the estimator's misspecified probability limit.

The remainder of the paper is organized as follows. Section~\ref{sec:frame} defines misspecification-robust sensitivity and informativeness in terms of the conditional expectation of the estimator given the moments, and expresses both measures through influence functions. Section~\ref{sec:sensgmm} derives the misspecification-robust influence functions for one-step, two-step, and iterated GMM (Section~\ref{sec:ifgmm}). It then characterizes the influence function of optimal minimum distance, in which the optimal weight's misspecification channel appears in isolation, and the implications for diagonal weighting (Proposition~\ref{prp:eff_channel}). Section~\ref{sec:discussions} illustrates the resulting sensitivity and informativeness through canonical examples, including 2SLS and a normal mean with a misspecified variance restriction. Section~\ref{sec:app} applies the diagnostics to the BLP model of automobile demand, the BPP model of household consumption insurance, and the AJRY income-and-democracy regressions. Section~\ref{sec:conclusion} concludes.

\section{Framework}\label{sec:frame}

\subsection{Setup and Definitions}\label{sec:setup}

Consider i.i.d.\ random vectors \(X_1,\ldots,X_n\) with unknown distribution \(P_0\).
Let \(g(X_i,\theta)\) be a \(q\times 1\) vector of moment functions, and let
\(\theta \in \Theta \subset \mathbb{R}^p\) be a \(p\times 1\) parameter vector with
\(q>p\). A GMM estimator is defined as
\[
    \widehat\theta
    = \arg\min_{\theta \in \Theta}
    \widehat g(\theta)' W \widehat g(\theta),
\]
where \(\widehat g(\theta) = n^{-1}\sum_{i=1}^n g(X_i,\theta)\) is the sample moment
vector and \(W\) is a positive semidefinite weight matrix. For expositional
clarity, we treat the weight matrix \(W\) as deterministic in this section and
consider stochastic weight matrices in Section~\ref{sec:ifgmm}. The moment condition
model is said to be correctly specified if there exists \(\theta \in \Theta\)
such that \(\mathbb{E}[g(X_i,\theta)] = 0\); otherwise the model is misspecified.

We denote by \(\theta_0\) the probability limit of the GMM estimator
\(\widehat\theta\). Under standard regularity conditions, \(\theta_0\) is the
unique minimizer of the population GMM criterion
\[
    Q(\theta) = \mathbb{E}[g(X_i,\theta)]' W \mathbb{E}[g(X_i,\theta)] .
\]
When the model is correctly specified, this minimizer satisfies
\(\mathbb{E}[g(X_i,\theta_0)] = 0\) and therefore coincides with the true
parameter. When the model is misspecified, \(\theta_0\) is the
pseudo-true parameter that best fits the moment conditions in the GMM
objective. Let \(g(\theta)=\mathbb{E}[g(X_i,\theta)]\) and \(g = g(\theta_0)\).

Suppose the standard regularity conditions hold, e.g., Theorem~1 of \citet*{hall2003large}. Then,
\begin{equation}
    \label{eq:asymdist}
    \sqrt{n}
    \begin{pmatrix}
        \widehat\theta - \theta_0 \\
        \widehat g(\theta_0) - g
    \end{pmatrix}
    \xrightarrow{d}
    \begin{pmatrix}
        \widetilde\theta \\
        \widetilde g
    \end{pmatrix}
    \sim N(0,\Sigma), \ \text{with } \Sigma =
    \begin{pmatrix}
        \sigma_{\theta \theta} & \sigma_{\theta g} \\ \sigma_{g \theta} & \sigma_{g g}
    \end{pmatrix}.
\end{equation}

By properties of the multivariate normal distribution, the conditional expectation of \(\widetilde\theta\) given \(\widetilde g\) is linear and given by
\begin{equation}
    \mathbb{E}[\widetilde\theta \mid \widetilde g]
    = \sigma_{\theta g}\sigma_{gg}^{-1}\widetilde g.
    \label{eq:condexp}
\end{equation}
We interpret the corresponding population regression coefficient as a measure of the sensitivity of the estimator to the moments.

\begin{definition}[Misspecification-Robust Sensitivity]
    \label{def:mrs}
    The misspecification-robust sensitivity (MRS) of the GMM estimator \(\widehat\theta\) with respect to the moment vector \(\widehat g(\theta_0)\) is defined as
    \[
        \Lambda = \sigma_{\theta g}\sigma_{gg}^{-1}.
    \]
\end{definition}

The sensitivity matrix \(\Lambda\) is \(p\times q\). Its \((k,l)\) element measures how a local deviation in the \(l\)th moment affects the first-order behavior of the \(k\)th component of the estimator around its probability limit, holding the other moments fixed. The linear relationship in equation~\eqref{eq:condexp} also appears in \citet*{andrews2017measuring}, who study sensitivity under correct specification and analyze the response of \(\widehat\theta\) to local perturbations of the data distribution. Our framework extends their sensitivity measure by allowing for misspecified moment conditions. Under correct specification, MRS nests the AGS sensitivity matrix (Section~\ref{sec:ifgmm}).

Sensitivity captures how the estimator responds locally to deviations in the realized moments. A complementary concept concerns the extent to which the moments are informative for the estimator. We adapt the notion of informativeness from \citet*{andrews2020informativeness} to the GMM setting under misspecification. In their framework, $\widehat c$ is a scalar structural estimator and $\widehat\gamma$ is a vector of descriptive statistics. Informativeness is the population $R^2$ from regressing $\widehat c$ on $\widehat\gamma$ under their joint asymptotic distribution. We specialize their setup by taking $\widehat c$ to be a component of the GMM estimator $\widehat\theta$ and $\widehat\gamma$ to be the sample moment vector $\widehat g(\theta_0)$.

\begin{definition}[Informativeness of Moments]
    \label{def:info}
    The informativeness of the moment vector for the \(k\)th component of the GMM estimator, \(\widehat\theta_k\), is defined as
    \[
        \Delta_k
        = \frac{\sigma_{\theta_k g}\sigma_{gg}^{-1}\sigma_{g\theta_k}}
        {\sigma_{\theta_k\theta_k}},
    \]
    where \(\sigma_{\theta_k\theta_k}\) denotes the \(k\)th diagonal element of \(\sigma_{\theta\theta}\) and \(\sigma_{\theta_k g}\) is the \(k\)th row of \(\sigma_{\theta g}\).
\end{definition}

By Definition~\ref{def:mrs}, $\sigma_{\theta_k g}\sigma_{gg}^{-1}$ is the $k$th row of $\Lambda$, so $\Delta_k$ in Definition~\ref{def:info} is the population $R^2$ from the linear projection of $\widetilde\theta_k$ on $\widetilde g$. We interpret $\Delta_k$ as a measure of structural efficiency under misspecification. By ``structural efficiency'' we mean the proportion of the estimator's asymptotic variance attributable to sampling variation in the moments, as opposed to variation arising from the Jacobian or the estimated weight matrix. Lower values indicate that a larger share of the estimator's variance is left unexplained by the moments.

Under correct specification, the influence function of the GMM estimator is $\psi(X_i) = \Lambda_{AGS}\, g(X_i, \theta_0)$, where $\Lambda_{AGS} = -(G'WG)^{-1}G'W$ and $G = \mathbb E[\partial g(X_i, \theta_0)/\partial \theta']$ is the population Jacobian of the moments at $\theta_0$. Since the influence function is a deterministic linear combination of the moments, it is spanned by the moments themselves, and $\Delta_k = 1$ for every $k$. \citet*{andrews2020informativeness} make essentially this observation in their footnote~2 (p.~2234) and implement it in their Section~5.1, Step~4: when the descriptive statistic is a vector of estimation moments that fully determines the structural estimator, $\Delta_k = 1$ by construction. Under misspecification, the GMM influence function contains additional terms that are not in the linear span of the moments, and $\Delta_k$ is generally less than one. We derive these additional terms in Section~\ref{sec:ifgmm}.

Structural efficiency is distinct from the test of overidentifying restrictions. The $J$-test asks whether the population moments are jointly compatible with zero at some parameter value. Structural efficiency instead asks how much of the estimator's asymptotic variation is linearly explained by sampling variation in the moments themselves. A rejected $J$-test does not by itself indicate whether the misspecification materially affects the precision or stability of a particular parameter estimate; $\Delta_k$ measures this channel. The two diagnostics are therefore complementary.

\subsection{Influence Function Representations of Sensitivity and Informativeness}\label{sec:ifsens}

To implement the proposed sensitivity and informativeness measures, we require estimates of the asymptotic covariance matrix \(\Sigma\). We use influence function representations for this purpose. For an asymptotically linear estimator such as the GMM estimator \(\widehat\theta\), an influence function \(\psi(x)\) satisfies
\begin{align}
     & \sqrt{n}(\widehat\theta - \theta_0) = \frac{1}{\sqrt{n}}\sum_{i=1}^{n}\psi(X_i) + o_p(1), \label{eq:theta_if} \\
     & \mathbb{E}[\psi(X_i)] = 0, \qquad
    \mathbb{E}[\psi(X_i)\psi(X_i)'] < \infty \nonumber .
\end{align}

Let \(\psi\) and \(\nu\) denote influence functions associated with the GMM estimator \(\widehat\theta\) and the sample moment vector \(\widehat g(\theta_0)\), respectively. Then
\begin{equation}
    \label{eq:expan}
    \sqrt{n}
    \begin{pmatrix}
        \widehat\theta - \theta_0 \\
        \widehat g(\theta_0) - g
    \end{pmatrix}
    =
    \frac{1}{\sqrt{n}}\sum_{i=1}^n
    \begin{pmatrix}
        \psi(X_i) \\
        \nu(X_i)
    \end{pmatrix}
    + o_p(1),
\end{equation}
where
\[
    \nu(X_i)
    = g(X_i,\theta_0) - g.
\]
As a consequence, the misspecification-robust sensitivity defined in Definition~\ref{def:mrs} can be written as
\begin{equation}
    \label{eq:iflambda}
    \Lambda
    = \mathbb{E}[\psi \nu']
    \mathbb{E}[\nu \nu']^{-1}.
\end{equation}
Similarly, the informativeness of the moments for the \(k\)th parameter, \(\Delta_k\) in Definition~\ref{def:info}, is given by
\begin{equation}
    \label{eq:ifdelta}
    \Delta_k
    = \frac{
        \mathbb{E}[\psi_k \nu']
        \mathbb{E}[\nu \nu']^{-1}
        \mathbb{E}[\nu \psi_k]
    }{
        \mathbb{E}[\psi_k^2]
    },
\end{equation}
where \(\psi_k\) is the \(k\)th component of \(\psi\).

A closely related expression appears in \citet*{iskrev2019expect}, who measures the sensitivity of structural parameters to calibrated parameters by treating the latter as if estimated. \citet*{jorgensen2023sensitivity} argues that calibrated parameters are non-stochastic and should be held fixed, leading to a different sensitivity measure. Equation~\eqref{eq:iflambda} concerns sample moments, which are stochastic by construction, and so aligns formally with the Iskrev treatment though applied to a different object.

Equations~\eqref{eq:iflambda} and \eqref{eq:ifdelta} are expressed in terms of population moments and influence functions. In practice, the sensitivity matrix \(\Lambda\) and the informativeness measure \(\Delta_k\) can be consistently estimated using a plug-in approach based on estimated influence functions. We construct estimated influence functions, denoted \(\widehat{\psi}_i\) and \(\widehat{\nu}_i\), for each observation \(i=1,\ldots,n\), by replacing unknown population quantities with their sample counterparts evaluated at the GMM estimate \(\widehat{\theta}\).

First, the estimated influence function for the moments, \(\widehat{\nu}_i\), is obtained by centering the moment function evaluated at the GMM estimate,
\[
    \widehat{\nu}_i
    = g(X_i,\widehat{\theta}) - \widehat g(\widehat{\theta}).
\]
Second, the estimated influence function for the GMM estimator, \(\widehat{\psi}_i\), is obtained by substituting sample estimators into the analytic expressions for the influence function \(\psi(X_i)\) derived in Proposition~\ref{prp:gmm_ifs} below. The population influence function depends on the Jacobian of the moment conditions, the weight matrix, and additional terms involving derivatives of the GMM criterion. We replace these population quantities with their sample counterparts evaluated at \(\widehat\theta\) and compute
\[
\widehat{\psi}_i = \widehat\psi(X_i;\widehat\theta).
\]
With these estimated influence functions in hand, the sensitivity matrix \(\widehat{\Lambda}\) can be computed using a sample analogue of equation~\eqref{eq:iflambda}. Specifically, \(\widehat{\Lambda}\) is obtained by regressing \(\widehat{\psi}_i\) on \(\widehat{\nu}_i\), which yields
\[
    \widehat{\Lambda}
    = \left( \sum_{i=1}^n \widehat{\psi}_i \widehat{\nu}_i^{\prime} \right)
    \left( \sum_{i=1}^n \widehat{\nu}_i \widehat{\nu}_i^{\prime} \right)^{-1}.
\]

Similarly, the informativeness measure \(\widehat{\Delta}_k\) can be estimated using a sample analogue of equation~\eqref{eq:ifdelta}, based on the sample variances and covariances of \(\widehat{\psi}_{ik}\) and \(\widehat{\nu}_i\). This implementation requires only the analytic form of the influence function \(\psi\), which is derived in the next section for a range of GMM estimators.

\section{Sensitivity Measures for GMM Estimators}\label{sec:sensgmm}

\subsection{Influence Function Representations under Misspecification}\label{sec:ifgmm}

Equations~\eqref{eq:iflambda} and \eqref{eq:ifdelta} show that the misspecification-robust sensitivity, \(\Lambda\), and informativeness, \(\Delta_k\), can be computed directly from the influence functions of the estimator and the moments. In this section, we derive the misspecification-robust influence functions for GMM estimators. As shown in \citet*{hall2003large} and subsequent work, the asymptotic behavior of GMM under misspecification depends on both the choice of the weight matrix and the estimation procedure.

We consider a class of GMM estimators defined as minimizers of
\[
    \widehat Q_n(\theta)
    = \widehat g(\theta)' \widehat W \widehat g(\theta),
\]
with different choices of the weight matrix.

In one-step GMM, the weight matrix does not depend on the parameter. We consider both the case of a deterministic weight matrix, where \(\widehat W = W\) is non-stochastic, and the case of an estimated weight matrix satisfying \(\widehat W \xrightarrow{p} W\), where \(\widehat W\) is asymptotically linear and does not depend on \(\theta\).

In two-step efficient GMM, the weight matrix is evaluated at a preliminary estimator \(\widehat\phi\). Specifically, the second-step estimator minimizes the criterion using the efficient weight matrix
\[
    \widehat W(\widehat\phi)
    = \left( \frac{1}{n}\sum_{i=1}^{n}
    g(X_i,\widehat\phi) g(X_i,\widehat\phi)' \right)^{-1}.
\]

Iterated GMM updates the weight matrix repeatedly by re-evaluating it at the current parameter estimate until convergence. At convergence, the estimator minimizes the criterion with weight matrix \(\widehat W(\widehat\theta)\).

We now introduce notation used in the influence function derivations. Let $\widehat W(\theta)$ be a parameter-dependent weight matrix estimator, and its probability limit be $W(\theta)$. Let $W = W(\theta_0)$.

The influence function derivations rely on the population curvature of the GMM objective. Let
\[
    R = \mathbb E\!\left[\frac{\partial}{\partial\theta'}\vect(G(X_i,\theta_0)')\right],
    \qquad H = (g'W\otimes I_p)R.
\]
These matrices are part of the one-step GMM influence function. Define
\[
    \gamma_i = \vect\{G(X_i,\theta_0)' - G'\}.
\]
We use the convention $\vect(G')$ rather than $\vect(G)$, under which Proposition~\ref{prp:gmm_ifs} relies on the identity $(G(X_i,\theta_0) - G)'Wg = (g'W\otimes I_p)\gamma_i$ for any weight matrix $W$. The influence functions of two-step and iterated GMM additionally involve
\[
    B = (g'\otimes G')S,
\]
where $S$ is the derivative of $\vect W(\cdot)$ with respect to its argument, evaluated at the relevant parameter value: $S = \partial \vect W(\phi)/\partial\phi'|_{\phi_0}$ for two-step GMM and $S = \partial \vect W(\theta)/\partial\theta'|_{\theta_0}$ for iterated GMM. The matrix $B$ captures the dependence of the first-order condition on the parameter inside the weight matrix, and governs the linearization of the iteration map (Proposition~\ref{prp:sens_to_1s}).

The matrices $H$ and $B$ depend on the misspecification vector $g$ and vanish under correct specification. The population curvature matrix is defined as
\[
    A = G'WG + H.
\]

We assume the following regularity conditions.
\begin{assumption}
    \label{ass:regularity}
    \begin{enumerate}[label=(\roman*)]

        \item
              The observations \(\{X_i\}_{i=1}^n\) are i.i.d.

        \item
              % The parameter space $\Theta\subset\mathbb R^p$ is compact.
              % The population GMM objective $Q(\theta)$ has a unique minimizer
              % $\theta_0$ in the interior of $\Theta$.
              % All population weight matrices $W$ are positive semi-definite.
              The parameter space \(\Theta \subset \mathbb{R}^p\) is compact.
              For one-step, two-step, and iterated GMM, the population objective \(Q(\theta)\) has a unique minimizer \(\theta_0\) in the interior of \(\Theta\).

        \item
              The moment function \(g(X,\theta)\) is three times continuously differentiable
              in \(\theta\in\Theta\) almost surely and satisfies the uniform integrability
              conditions
              \[
                  \mathbb E\!\left[\sup_{\theta\in\Theta}\|g(X,\theta)\|^2\right]<\infty,
                  \qquad
                  \mathbb E\!\left[\sup_{\theta\in\Theta}\|G(X,\theta)\|^2\right]<\infty,
              \]
              for all \(j_1,j_2\in\{1,\dots,p\}\),
              \[
                  \mathbb E\!\left[
                      \sup_{\theta\in\Theta}
                      \left\|
                      \frac{\partial^2}{\partial\theta_{j_1}\partial\theta_{j_2}}
                      g(X,\theta)
                      \right\|^2
                      \right]<\infty,
              \]
              and for all \(j_1,j_2,j_3\in\{1,\dots,p\}\),
              \[
                  \mathbb E\!\left[
                      \sup_{\theta\in\Theta}
                      \left\|
                      \frac{\partial^3}{\partial\theta_{j_1}\partial\theta_{j_2}\partial\theta_{j_3}}
                      g(X,\theta)
                      \right\|
                      \right]<\infty.
              \]
              % These conditions ensure uniform laws of large numbers for
              % $\widehat g(\theta)$, $\widehat G(\theta)$, and $\widehat R(\theta)$
              % over $\Theta$.

        \item
              The population curvature matrices \(A\), \(A+B\) are nonsingular.

        \item
              For iterated GMM, the population updating map is a contraction at
              \(\theta_0\).

        \item
              Weight matrix conditions:
              \begin{enumerate}
                  \item[(a)] (Estimated weight, one-step GMM) The weight matrix $\widehat W$ does not depend on $\theta$, $\widehat W \xrightarrow{p} W$ where $W$ is symmetric positive definite, and
                        \[
                            \sqrt n\,\vect(\widehat W - W) = \frac{1}{\sqrt n}\sum_{i=1}^n \omega_i + o_p(1), \quad \mathbb E[\|\omega_i\|^2] < \infty.
                        \]
                  \item[(b)] (Parameter-dependent weight) Let $\widehat\phi$ be a preliminary estimator converging in probability to $\phi_0$, with $\sqrt n(\widehat\phi - \phi_0) = n^{-1/2}\sum_i \psi^\phi(X_i) + o_p(1)$. There exists a neighborhood $\mathcal N$ of $\phi_0$ such that $W(\phi)$ is continuously differentiable and $W(\phi_0)$ is positive definite. With probability approaching one, $\widehat W(\phi)$ is continuously differentiable on $\mathcal N$ and satisfies
                        \[
                            \sup_{\phi\in\mathcal N}\|\widehat W(\phi) - W(\phi)\| \xrightarrow{p} 0, \qquad
                            \sup_{\phi\in\mathcal N}\left\|\frac{\partial \vect(\widehat W(\phi))}{\partial \phi'} - \frac{\partial \vect(W(\phi))}{\partial \phi'}\right\| \xrightarrow{p} 0,
                        \]
                        and the influence function $\omega_i$ of $\widehat W(\phi_0)$ satisfies
                        \[
                            \sqrt n\,\vect(\widehat W(\phi_0) - W(\phi_0)) = \frac{1}{\sqrt n}\sum_{i=1}^n \omega_i + o_p(1), \quad \mathbb E[\|\omega_i\|^2] < \infty.
                        \]
              \end{enumerate}

    \end{enumerate}
\end{assumption}

Under Assumption~\ref{ass:regularity}, the GMM estimators considered above are consistent for their pseudo-true probability limits and asymptotically linear, with influence functions given by Proposition~\ref{prp:gmm_ifs}. Consistency under misspecification follows from standard arguments for misspecified GMM \citep*[e.g.,][]{hall2003large,hansen2021inference,hwang2022doubly}; these influence functions are the inputs to the misspecification-robust sensitivity and informativeness measures in equations~\eqref{eq:iflambda} and~\eqref{eq:ifdelta}.

\begin{proposition}[Decomposition of GMM Influence Functions under Misspecification]
    \label{prp:gmm_ifs}
    Suppose Assumption~\ref{ass:regularity} holds. The influence functions of the GMM estimators in Section~\ref{sec:ifgmm} are decomposed as follows.
    \noindent\textnormal{(i)} \emph{One-step GMM with deterministic weight $W$:}
    \[
        \psi^{1s}(X_i) = -A^{-1}\!\left[G'W\nu_i + (g'W\otimes I_p)\gamma_i\right].
    \]
    \noindent\textnormal{(ii)} \emph{One-step GMM with estimated weight $\widehat W$:} suppose Assumption~\ref{ass:regularity}(vi)(a) holds. Then
    \[
        \psi^{1s,W}(X_i) = -A^{-1}\!\left[G'W\nu_i + (g'W\otimes I_p)\gamma_i + (g'\otimes G')\omega_i\right].
    \]
    \noindent\textnormal{(iii)} \emph{Two-step efficient GMM:} suppose Assumption~\ref{ass:regularity}(vi)(b) holds and let $W = W(\phi_0)$. Then
    \[
        \psi^{2s}(X_i) = -A^{-1}\!\left[G'W\nu_i + (g'W\otimes I_p)\gamma_i + (g'\otimes G')\omega_i + B\psi^\phi(X_i)\right].
    \]
    \noindent\textnormal{(iv)} \emph{Iterated GMM:} let $W = W(\theta_0)$. Then
    \[
        \psi^{it}(X_i) = -(A+B)^{-1}\!\left[G'W\nu_i + (g'W\otimes I_p)\gamma_i + (g'\otimes G')\omega_i\right].
    \]
\end{proposition}

Each bracket contains a moment channel ($G'W\nu_i$) and misspecification channels (Jacobian, weight-matrix, and first-step). For ease of reference, let
\[
    M_\nu = -A^{-1}G'W, \quad M_\gamma = -A^{-1}(g'W\otimes I_p), \quad M_\omega = -A^{-1}(g'\otimes G').
\]
The iterated case differs from the one-step estimated-weight case only in replacing $A^{-1}$ with $(A+B)^{-1}$. The matrices $M_\gamma$, $M_\omega$, and $B$ each contain $g$ as a factor and vanish identically when $g = 0$.

\begin{remark}
    \label{rem:agsbias}
    The moment-channel coefficient $M_\nu$ is the asymptotic bias of the one-step GMM pseudo-true value under an additive shift of the moment level. With the weight matrix $W$ held fixed, the pseudo-true value under the shifted moment $g(\theta)+\delta$ solves the first-order condition
    \[
        \mathbb E[G(X_i,\theta)]'W\bigl(g(\theta)+\delta\bigr)=0 .
    \]
    Differentiating in $\theta'$ at $\theta_0$ and $\delta=0$ gives $G'WG$ from the moment factor and, from differentiating the Jacobian factor against the moment value $g$, the curvature term $H=(g'W\otimes I_p)R$, which together form $A=G'WG+H$; differentiating in $\delta'$ gives $G'W$. The implicit function theorem then yields
    \[
        \frac{\partial\theta_0}{\partial\delta'}=-A^{-1}G'W=M_\nu .
    \]
    AGS measure sensitivity at a correctly specified baseline through a local perturbation that vanishes at rate $n^{-1/2}$; there the moment value is zero, $H$ vanishes, and the bias coefficient is $\Lambda_{AGS}=-(G'WG)^{-1}G'W$. They note that under a fixed alternative the pseudo-true value shifts by approximately $\Lambda_{AGS}$ times the induced moment change, but do not adopt this fixed-misspecification regime, because the deviation then grows large relative to the sampling error as the sample size increases. We work in exactly that regime: at the misspecified pseudo-true value $g\neq0$, so the coefficient is instead $M_\nu=-(G'WG+H)^{-1}G'W$, which augments the curvature matrix $G'WG$ of $\Lambda_{AGS}$ by $H$, the population counterpart of the curvature term carried by the finite-sample sample-sensitivity matrix of AGS and driven to zero by their correct-specification asymptotics. The two coincide when $H=0$, in particular under zero moment curvature ($R=0$) such as in linear instrumental variables. A general local deviation of the data-generating process also moves $G$ and $W$, activating the Jacobian and weight-matrix channels, so the bias is then governed by the full misspecification-robust sensitivity $\Lambda$ rather than by $M_\nu$ alone.
\end{remark}

Versions of these influence functions appear in \citet*{hall2003large}, \citet*{hwang2022doubly}, and \citet*{hansen2021inference}. Continuously Updating GMM (CUGMM), which jointly updates the parameter and the weight matrix, presents distinct theoretical challenges under misspecification, including non-convexity and the risk of variance collapse \citep*{kleibergen2025double}. We provide analysis of CUGMM, including its consistency and a derivation of the misspecification-robust influence function, in Appendix~\ref{sec:cugmm}.

Under $g = 0$, $M_\gamma = M_\omega = 0$, so $\psi(X_i) = M_\nu\nu_i$, $\Lambda$ collapses to $M_\nu = \Lambda_{AGS}$, and $\Delta_k = 1$ for every $k$. The condition $g = 0$ is sufficient for $\Delta_k = 1$ but not necessary. $\Delta_k$ can equal one under misspecification when the misspecification channels lie in the linear span of the moments. The next corollary states the necessary and sufficient condition.

\begin{corollary}[Misspecification and Structural Efficiency]
    \label{cor:miss}
    Suppose $\mathbb E[\psi_k(X_i)^2]>0$. For each estimator in Proposition~\ref{prp:gmm_ifs}, write
    \[
        \psi_k(X_i) = (M_\nu)_k\nu_i + r_k(X_i),
    \]
    where $(M_\nu)_k$ is the $k$th row of $M_\nu$ and $r_k(X_i)$ collects all misspecification-channel terms. Let $\mathcal L(\nu_i)$ denote the closed linear span of $\nu_i$ in $L^2(P)$ and $\Pi_\nu^\perp$ the orthogonal projection onto its complement. Then
    \[
        \Delta_k = 1 - \frac{\mathbb E[(\Pi_\nu^\perp r_k(X_i))^2]}{\mathbb E[\psi_k(X_i)^2]}.
    \]
    In particular, $\Delta_k = 1$ if and only if $r_k(X_i)\in\mathcal L(\nu_i)$, and $\Delta_k < 1$ if and only if $\mathbb E[(\Pi_\nu^\perp r_k(X_i))^2]>0$.
\end{corollary}

Among the channels of Proposition~\ref{prp:gmm_ifs}, the weight-matrix channel is the one the researcher controls through the choice of weight matrix. This choice is most transparent in minimum distance estimation, where the moment is the matching residual $g(X_i,\theta)=m_i-f(\theta)$: $m_i$ collects the data moments, with mean $\mu$ and covariance $V=\operatorname{Var}(m_i)$, and $f$ is the model-implied map. Minimum distance isolates the weight-matrix channel. The centered moment $\nu_i=m_i-\mu$ and the covariance $V$ do not depend on $\theta$, so the optimal weight $\widehat W=\widehat V^{-1}$, the inverse of the sample covariance of the data moments, involves no preliminary parameter estimate and the first-step channel is absent; the Jacobian $-\partial f(\theta)/\partial\theta'$ is nonrandom, so $\gamma_i=0$ and the Jacobian channel vanishes. Only the moment and weight-matrix channels remain, and the influence function of the optimal minimum distance (OMD) estimator can be characterized in full. Minimum distance is also the setting where the choice between the optimal and a diagonal weight is a long-standing practical concern \citep*{altonji1996small,cheng2026weight}, and the setting of the BPP application in Section~\ref{sec:bpp}. Proposition~\ref{prp:eff_channel} derives the OMD influence function in closed form and relates the resulting variance and informativeness to the population overidentification criterion $g'Wg$.

\begin{proposition}[Optimal Minimum Distance]
    \label{prp:eff_channel}
    Suppose Assumption~\ref{ass:regularity}(vi)(a) holds. The influence function
    of the OMD estimator is
    \[
        \psi(X_i)=M_\nu\nu_i\,(1-\nu_i'Wg).
    \]
    Let
    $V_0=\operatorname{Var}(M_\nu\nu_i)=A^{-1}G'WGA^{-1}$ denote the asymptotic
    variance under the fixed weight $W$, and $V_{0,kk}=\operatorname{Var}[(M_\nu\nu_i)_k]$
    its $k$th diagonal entry. Then, for each coordinate $k$,
    \[
        \operatorname{Var}[\psi_k(X_i)]=(1+g'Wg)\,V_{0,kk}+\kappa_k,\qquad
        \kappa_k=\operatorname{Cov}\!\big((M_\nu\nu_i)_k^2,\,(1-\nu_i'Wg)^2\big),
    \]
    and $\Delta_k\le1$, with equality if $g=0$. If $\nu_i$ is Gaussian,
    $\kappa_k=0$ and
    \[
        \operatorname{Var}[\psi(X_i)]=(1+g'Wg)\,V_0,\qquad
        \Delta_k=\frac{1}{1+g'Wg}\quad\text{for every }k,
    \]
    uniformly below one under misspecification.
\end{proposition}

The weight-matrix channel is the moment channel times the mean-zero scalar $\nu_i'Wg$, a product quadratic in $\nu_i$ that lies outside $\mathcal L(\nu_i)$ and, by Corollary~\ref{cor:miss}, lowers $\Delta_k$. The squared multiplier $(1-\nu_i'Wg)^2$ has mean exactly $1+g'Wg$. Relative to the benchmark $V_0$, estimating the optimal weight adds the population overidentification criterion times the moment-channel variance, plus the cumulant $\kappa_k$. Under a fixed weight, e.g., equally weighted minimum distance (EWMD) which fixes the weight at the identity, the channel is absent, $\psi(X_i)=M_\nu\nu_i$, and $\Delta_k=1$ regardless of misspecification. The gap is therefore the informativeness cost of estimating the efficient weight. This structure is specific to weights that invert the covariance matrix of the moments. A generic data-dependent weight has a weight-matrix channel too, but one that neither factors through a single scalar nor scales with the population overidentification criterion.

The diagonally weighted minimum distance (DWMD) estimator is another popular choice in practice. The DWMD weight is $\widehat D=(\operatorname{diag}\widehat V)^{-1}$, the inverse of the diagonal of the same sample covariance that OMD inverts in full: each moment is weighted by the inverse of its own sampling variance, $\widehat D_{kk}=1/\widehat V_{kk}$.\footnote{This inverse-variance weighting follows \citet*{chatterjee2021estimating} and \citet*{cheng2026weight}. It is distinct from $\operatorname{diag}(\widehat V^{-1})$, the diagonal of the optimal weight, which coincides with it only when $V$ is diagonal; either replaces the rank-one outer product with a full-rank diagonal.} Estimating this diagonal weight contributes a weight-matrix channel through the influence function $\omega_i^D$ of $\widehat D$, but one that does not simplify as in OMD. The $k$th diagonal entry $1/\widehat V_{kk}$ has influence function $-(\nu_{ik}^2-V_{kk})/V_{kk}^2$, where $\nu_{ik}$ is the $k$-th coordinate of $\nu_i$, so the channel is a sum of per-moment terms,
\[
    (g'\otimes G')\,\omega_i^D
    = -\sum_{k=1}^{q} G_{k\cdot}'\,g_k\,\frac{\nu_{ik}^2-V_{kk}}{V_{kk}^2}
    = -\sum_{k=1}^{q} G_{k\cdot}'\,g_k\,\frac{\nu_{ik}^2}{V_{kk}^2},
\]
where $G_{k\cdot}$ is the $k$th row of $G$ and the second equality uses the DWMD (population) first-order condition $G'(\operatorname{diag}V)^{-1}g=0$. Each term is quadratic in a single coordinate $\nu_{ik}$ rather than the scalar $\nu_i'Wg$ behind the OMD factorization. It neither factors through one scalar nor scales with the population overidentification criterion $g'Wg$: by keeping only the diagonal of the weight, DWMD discards the cross-moment information that the efficient weight uses.

\begin{remark}
    \label{rem:altonji}
    Applied researchers frequently estimate with a non-optimal weight matrix such as 
    equal weighting, diagonal weighting, or one-step GMM. \citet*{altonji1996small} found that EWMD often outperforms
    OMD in finite samples and attributed the gap to sampling error in the
    estimated weight. Proposition~\ref{prp:eff_channel} provides an asymptotic
    counterpart of their finding: the optimal weight contributes a weight-matrix channel whose variance is the overidentification criterion times the moment-channel variance (exactly so under Gaussian moments, up to the cumulant $\kappa_k$ in general), while a fixed weight contributes none. The practice trades
    efficiency for informativeness, and the trade is favorable under
    misspecification, where the optimal weight is in any case not variance
    minimizing. The simpler estimator does not dominate, however: it is typically less precise. In the BPP application of Section~\ref{sec:bpp}, the
    misspecification-robust standard error of $\phi$ rises from $0.043$ under
    OMD to $0.113$ under DWMD, and the point estimates diverge, which is itself evidence of misspecification. The simpler weight preserves informativeness, with $\Delta\approx 1$, and forgoes only an efficiency gain that the optimal weight no longer delivers.
\end{remark}

\smallskip

\begin{remark}
    \label{rem:scaling}
    The OMD and DWMD channels scale differently with the number of moments $q$. Take $V = I_q$ with $\nu_i$ Gaussian, so the optimal weight is $W = I_q$ and coincides with the DWMD weight $(\operatorname{diag}V)^{-1}$ and the EWMD weight: the three estimators share the same population weight and differ only in which parts of it are estimated, the full matrix for OMD, the diagonal for DWMD, none for EWMD. Suppose each moment is misspecified by the same amount, $g_k = \delta$ for all $k$, giving $g'Wg = q\delta^2$. By Proposition~\ref{prp:eff_channel} (with $\kappa_k = 0$ under Gaussianity\footnote{Empirically the cumulant correction need not be small: in the BPP application of Section~\ref{sec:bpp}, fourth cumulants account for roughly $27\%$ of the variance of the weight-matrix channel for $\phi$.}), the OMD weight-matrix channel has variance $q\delta^2\,(G'G)_{kk}$; for DWMD the channel is $-G'(g\circ\nu_i\circ\nu_i)$ and by independence of $\nu_{ij}^2$ across $j$ has variance $2\delta^2\,(G'G)_{kk}$. Both channels are quadratic in $\nu_i$ and orthogonal to $\mathcal L(\nu_i)$, so they enter informativeness in full. The moment channel $G'\nu_i$ is common to the three estimators, so its variance $(G'G)_{kk}$ cancels in the ratio and
    \[
        \Delta_k^{OMD} = \frac{1}{1+q\delta^2},\qquad \Delta_k^{DWMD} = \frac{1}{1+2\delta^2},\qquad \Delta_k^{EWMD} = 1.
    \]
    The OMD channel grows with $q$ while the DWMD channel does not; in the BPP application of Section~\ref{sec:bpp}, $q = 325$, and the estimated informativeness for OMD is much lower than that of DWMD. 
\end{remark}

The influence function for two-step GMM derived in
Proposition~\ref{prp:gmm_ifs}(iii) has an intuitive decomposition. It consists
of a direct estimation effect and an indirect effect arising from the first-step
estimator \(\widehat{\phi}\). This structure allows us to formally characterize
the sensitivity of the second-step estimator to the first-step estimator.

\begin{proposition}[Sensitivity to the First-Step Estimator]
    \label{prp:sens_to_1s}
    Suppose Assumption~\ref{ass:regularity} holds. Let \(\widehat\phi\) be a
    first-step estimator with probability limit \(\phi_0\), and let
    \(\widehat\theta(\widehat\phi)\) denote the corresponding second-step GMM
    estimator defined as the solution to the sample first-order condition given
    \(\widehat\phi\). Then the sensitivity of the second-step estimator to the
    first-step estimator is given by
    \[
        \Lambda_{\phi}
        = \text{plim} \frac{\partial \widehat\theta(\widehat\phi)}{\partial
            \widehat\phi'}
        = -A^{-1}B,
    \]
    where \(A\) and \(B\) are evaluated at \((\theta_0,\phi_0)\).
\end{proposition}

Proposition~\ref{prp:sens_to_1s} shows that the indirect component of the two-step GMM influence function arises because first-step estimation uncertainty enters through the iteration map. The matrix \(A^{-1}B\) represents the linearization of a single iteration step, as in the two-step estimator. Iterated GMM repeatedly applies the same iteration map, so the local behavior of the iterated procedure is governed by successive applications of this operator. The contraction condition in Assumption~\ref{ass:regularity}(v) is imposed on the population iteration map underlying the two-step estimator. Linearizing this map around \(\theta_0\) yields the Jacobian \(-A^{-1}B\), so the contraction condition is equivalent to requiring \(\rho(-A^{-1}B)<1\), where \(\rho(\cdot)\) denotes the spectral radius.

Under this condition, the $s$-step influence functions follow a recursion
driven by repeated application of $A^{-1}B$. Applying
Proposition~\ref{prp:gmm_ifs}(iii) with the $(s-1)$-step estimator as the
first step,
\[
    \psi^{ss}(X_i) = -A_s^{-1}\big[f_{2,s}(X_i) + B_s\,\psi^{(s-1)(s-1)}(X_i)\big],
    \qquad \psi^{11}(X_i)=\psi^{1s,W}(X_i),
\]
where $f_{2,s}$ collects the moment, Jacobian, and weight-matrix channels of
Proposition~\ref{prp:gmm_ifs}(iii) and $A_s$, $B_s$, $f_{2,s}$ are evaluated
at the step-specific pseudo-true values $(\theta^{(s)},\theta^{(s-1)})$.
\citet*{hansen2021inference} establish that $\theta^{(s)}$ converges to the
fixed point $\theta_0$ geometrically; Lemma~\ref{lem:estimand} in the
Appendix states this result in the form used in the proof of
Proposition~\ref{prp:variance_reduction}. The next proposition shows that the
$s$-step influence function converges to the iterated influence function at
the same geometric rate, yielding a fixed-point interpretation of iterated
GMM under misspecification.

\begin{proposition}[Geometric convergence of iterated GMM]
    \label{prp:variance_reduction}
    Suppose Assumption~\ref{ass:regularity} holds, strengthened so that the
    weight influence function $\omega_i(\cdot)$ is locally Lipschitz in
    $L^2(P)$ and $S(\cdot)$ is locally Lipschitz near $\theta_0$, and that
    the first-step probability limit lies in the neighborhood $\mathcal N_0$
    of $\theta_0$ given in Lemma~\ref{lem:estimand}. For every
    $\tilde\rho\in(\rho(-A^{-1}B),1)$ there is a constant $C_{\tilde\rho}>0$
    such that, for all $s\ge1$,
    \[
        \|\psi^{ss}(X_i)-\psi^{it}(X_i)\|_{L^2(P)}
        \;\le\; C_{\tilde\rho}\,\tilde\rho^{\,s-1},
    \]
    where $\|V\|_{L^2(P)}=(\mathbb E\|V\|^2)^{1/2}$. Consequently
    $\operatorname{Var}(\psi^{ss})\to\operatorname{Var}(\psi^{it})$ and, for
    every $k$ with $\operatorname{Var}[\psi^{it}_k(X_i)]>0$,
    $\Delta_k^{ss}\to\Delta_k^{it}$, both at the rate $\tilde\rho^{\,s-1}$.
\end{proposition}

Propositions~\ref{prp:sens_to_1s} and~\ref{prp:variance_reduction} together provide two practical diagnostics. Proposition~\ref{prp:sens_to_1s} gives $\Lambda_\phi = -A^{-1}B$, the sensitivity of the two-step estimator's probability limit to perturbations in the first-step probability limit. A researcher using two-step GMM can compute $\Lambda_\phi$ at $\widehat\theta$ to assess how much the second-step estimate depends on the first-step choice. Proposition~\ref{prp:variance_reduction} makes the spectral radius $\rho(-A^{-1}B)$ the rate diagnostic: the $s$-step influence function, and with it the misspecification-robust variance and informativeness, approach their iterated limits geometrically at any rate above $\rho(-A^{-1}B)$. To leading order, each weight update shrinks the remaining gap by the factor $\rho$. At the values estimated in our iterating applications, $\rho=0.61$ in AJRY and $0.64$ in BLP (Sections~\ref{sec:ajry} and~\ref{sec:blp}), the gap falls below $10\%$ of its initial size within six and seven steps, respectively. In BPP the OMD weight is parameter-free, so $B=0$, $\rho=0$, and the iteration is trivial. Both diagnostics are computable from the same quantities used to construct misspecification-robust standard errors and require no additional estimation.

\subsection{Discussions on GMM Sensitivity}\label{sec:discussions}

\subsubsection{One-Step GMM}

When the weight matrix \(W\) is deterministic, the influence function is given by Proposition \ref{prp:gmm_ifs}(i):
\[
    \psi^{1s}(X_i) = -A^{-1} \left[ G'W\nu_i + (g'W\otimes I_p)\gamma_i \right] = M_\nu\nu_i + M_\gamma\gamma_i.
\]
This structure reveals how misspecification alters the sensitivity relative to the correctly specified case. The population curvature matrix \(A = G'WG + H\) now includes the term
\(H = (g'W \otimes I_p)R\), which arises from the curvature of the moment function evaluated at a parameter value where the population moments are nonzero.

Under misspecification, the influence function contains the Jacobian channel $M_\gamma\gamma_i$ and, when the weight matrix is estimated, the weight-matrix channel $M_\omega\omega_i$. By Corollary~\ref{cor:miss}, these channels reduce informativeness whenever their sum has a component orthogonal to $\mathcal L(\nu_i)$ in $L^2(P)$. Both channels are proportional to the misspecification vector $g$ and vanish under correct specification. They also vanish in one-step GMM with a deterministic weight matrix when the Jacobian is deterministic across observations ($\gamma_i = 0$): the influence function reduces to $\psi^{1s}(X_i) = M_\nu\nu_i$ and $\Delta_k = 1$ for every parameter, regardless of the Hansen $J$ statistic computed at the efficient weight. Minimum distance estimation \citep*{blundell2008consumption} provides the leading instance, examined in Section~\ref{sec:bpp}. The 2SLS estimator illustrates the case where the Jacobian channel is active, examined in Example~\ref{exm:2sls}.

\begin{example}[2SLS]
    \label{exm:2sls}
    Consider a linear IV model \(Y_i = D_i\theta + \epsilon_i\) with instruments \(Z_i\) and moment function \(g(X_i,\theta) = Z_i(Y_i - D_i\theta)\), estimated by 2SLS with weight \(\widehat W = (n^{-1}Z'Z)^{-1}\). This is one-step GMM with an estimated weight matrix, so Proposition~\ref{prp:gmm_ifs}(ii) gives
    \[
        \psi^{1s,W}(X_i) = M_\nu\nu_i + M_\gamma\gamma_i + M_\omega\omega_i.
    \]
    The moment is linear in \(\theta\), so the Jacobian \(G(X_i) = -Z_iD_i\) does not depend on \(\theta\), giving \(R = 0\), \(H = 0\), and \(A = G'WG\). The estimated weight matrix has influence function \(\omega_i = \vect(-W(Z_iZ_i' - \Omega)W)\), where \(\Omega = \mathbb E[Z_iZ_i']\) and \(W = \Omega^{-1}\). Linearity removes the curvature term \(H\), but it does not remove the misspecification channels. The Jacobian channel \(M_\gamma\gamma_i\) remains active because \(\gamma_i = \vect(G(X_i)' - G') \neq 0\) even though \(R = 0\). The weight channel \(M_\omega\omega_i\) remains active because the weight matrix is estimated. Both coefficients are proportional to \(g\) and vanish only when \(g = \mathbb E[Z_i\epsilon_i] = 0\), in which case \(\psi^{1s,W} = M_\nu\nu_i\), the sensitivity equals \(\Lambda_{AGS}\), and \(\Delta = 1\). Misspecification arises naturally under treatment-effect heterogeneity. With more instruments than parameters, no single \(\theta\) sets all moments to zero, so \(g \neq 0\) and the pseudo-true value is the GMM-weighted average of heterogeneous effects. The two channels are then nonzero, and the components of \(\gamma_i\) and \(\omega_i\) orthogonal to \(\mathcal L(\nu_i)\) drive informativeness below one by Corollary~\ref{cor:miss}. These channels coincide with the additional variance terms derived by \citet*{lee2018consistent} for 2SLS under heterogeneity. Both the sensitivity and the informativeness of 2SLS therefore depart from their correctly specified values.
\end{example}

\subsubsection{Efficient GMM Estimators}

Efficient GMM estimators use a weight matrix that depends on the parameter, \(W(\theta)=\Omega(\theta)^{-1}\) with \(\Omega(\theta)=\mathbb E[g(X_i,\theta)g(X_i,\theta)']\) the uncentered second-moment matrix. In the two-step GMM estimator, the weight matrix is evaluated at a preliminary estimator \(\widehat{\phi}\), so the second-step estimator inherits uncertainty from the first step. Proposition \ref{prp:gmm_ifs}(iii) shows that the influence function takes the form
\[
    \psi^{2s}(X_i) = -A^{-1}\!\left[G'W\nu_i + (g'W\otimes I_p)\gamma_i + (g'\otimes G')\omega_i \right]-A^{-1}B\psi^\phi(X_i).
\]
The matrix \(B\) captures the sensitivity of the optimal weight matrix to the parameter, and the term \(-A^{-1}B\,\psi^\phi(X_i)\) quantifies how first-step estimation uncertainty enters the second-step estimator. Under correct specification, \(g=0\) and hence \(B=0\), recovering the standard result that the choice of the first-step estimator does not affect asymptotic efficiency.

Iterated GMM repeatedly updates the weight matrix until the parameter used in the weight matrix coincides with the estimated parameter. As shown in Proposition \ref{prp:gmm_ifs}(iv), this alters the effective curvature of the problem from \(A\) to \(A+B\). The estimator no longer depends on an external preliminary estimator, but sensitivity to misspecification persists in the curvature through \(B\). Continuously updating GMM, which minimizes over the parameter inside the weight matrix as well, is treated in Appendix~\ref{sec:cugmm}.

The following example illustrates a central implication of these results: efficient GMM estimators can share the same local sensitivity while exhibiting substantially different informativeness under misspecification.

\begin{example}[Normal mean with misspecified variance restriction; \citealp{schennach2007point}]
    \label{exm:schennach}
    Consider i.i.d.\ data \(X_i\sim N(\mu,\sigma^2)\) and the overidentified moment
    vector
    \[
        g(X_i,\theta)
        =
        \begin{pmatrix}
            X_i-\theta \\
            (X_i-\theta)^2-1
        \end{pmatrix}.
    \]
    The parameter of interest is the mean \(\theta\), while the variance restriction may be misspecified when
    \(\sigma^2\neq1\).

    All three estimators have pseudo-true value \(\theta_0=\mu\) over \(\sigma^2\geq1/2\) (Table~\ref{tab:schennach}). The sensitivity is identical across the three estimators,
    \(\Lambda=(1,0)\), because each estimator locally tracks the sample mean.

    Despite identical sensitivity, informativeness differs substantially across estimators. When the weight matrix depends on the parameter, the influence function contains additional components arising from estimation of the weight matrix under misspecification. These components inflate the variance of the estimator without affecting local sensitivity, so informativeness satisfies \(\Delta<1\). Table~\ref{tab:schennach} summarizes the sensitivity and informativeness measures for one-step, two-step, and iterated estimators.

\end{example}

By Corollary~\ref{cor:miss}, the differences in informativeness across estimators arise from which misspecification channels lie outside $\mathcal L(\nu_i)$. For one-step GMM with $W=I$, the only active channel is the Jacobian channel: $r_k(X_i) = M_\gamma\gamma_i$, with $\gamma_i = (0,-2U_i)'$ and $\nu_i = (U_i, U_i^2-\sigma^2)'$, where $U_i = X_i - \mu$. The Jacobian channel reduces to a multiple of $U_i$, which is the first coordinate of $\nu_i$, so it adds no information beyond the moments themselves and $\Delta = 1$. For two-step and iterated GMM, the weight-matrix channel involves the influence function of the estimated second-moment matrix, which contains centered cubic and quartic terms in $U_i$ that are not in $\mathcal L(\nu_i)$; hence $\Delta<1$.

\begin{table}[ht]
    \centering
    \renewcommand{\arraystretch}{1.5}
    \begin{tabular}{lcc}
        \hline
        \textbf{Estimator} &  \textbf{Sensitivity} $\Lambda$ & \textbf{Informativeness} $\Delta$                                                        \\ \hline
        One-Step ($W=I$)                      & $(1, 0)$                       & $1$                                                                                      \\
        Two-Step GMM                               & $(1, 0)$                       & $\dfrac{(5\sigma^4-4\sigma^2+1)^2}{(5\sigma^4-4\sigma^2+1)^2 + 6\sigma^4(\sigma^2-1)^2}$ \\
        Iterated GMM                              & $(1, 0)$                       & $\dfrac{2\sigma^4}{5\sigma^4 - 6\sigma^2 + 3}$                                           \\ \hline
    \end{tabular}
    \caption{Sensitivity and Informativeness in \citet*{schennach2007point} Model}
    \label{tab:schennach}
\end{table}

\begin{figure}[ht]
    \centering
    \includegraphics[width=0.8\linewidth]{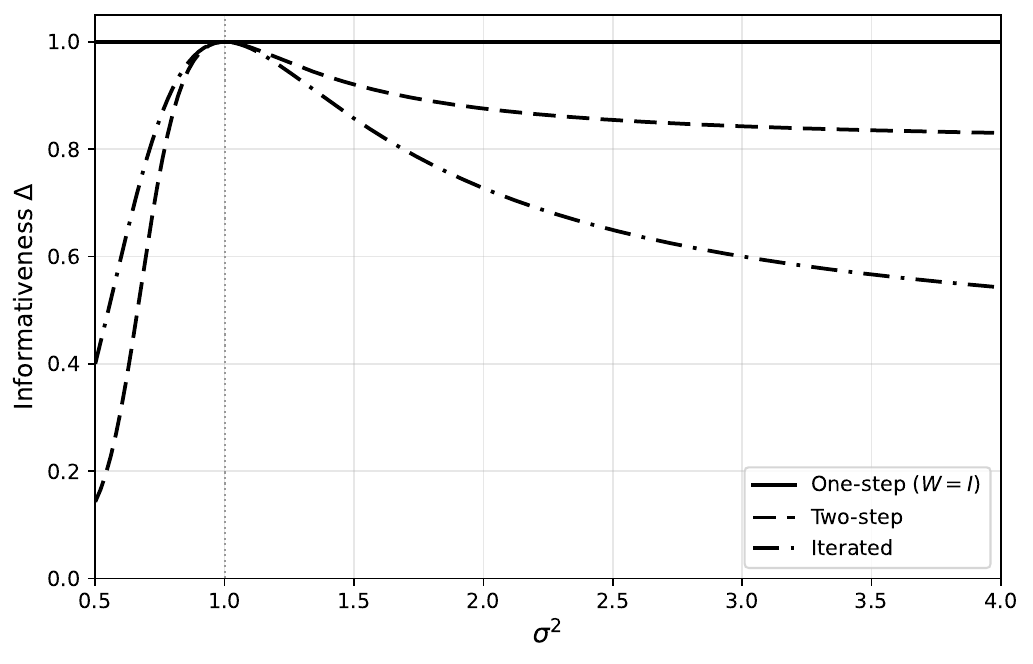}
    \caption{Informativeness of the influence function for one-step ($W=I$), two-step, and iterated GMM in the normal-mean model with misspecified variance restriction, plotted against $\sigma^2$ over the range $\sigma^2 \geq 1/2$.}
    \label{fig:schennach_panel}
\end{figure}
%The CUGMM analogue, derived in Appendix~\ref{sec:cugmm}, is omitted for clarity.

Figure~\ref{fig:schennach_panel} plots informativeness for the three estimators as a function of $\sigma^2$ over the range $\sigma^2 \geq 1/2$ in which all three pseudo-true values coincide with $\mu$. The one-step estimator with identity weight has $\Delta = 1$ throughout. The two efficient estimators coincide with the one-step estimator at $\sigma^2 = 1$, where the variance restriction is correctly specified, and fall below one elsewhere. Their informativeness ranking reverses across $\sigma^2 = 1$: two-step dominates iterated for $\sigma^2 > 1$, while iterated dominates two-step for $\sigma^2 < 1$.

\begin{example}[Schennach, Cont'd]
    \label{exm:schennach2}
    Example~\ref{exm:schennach} varies the weight matrix while holding the moment function fixed; here we do the reverse. Let $X_i \sim N(0,\sigma^2)$ and consider two researchers who incorrectly assume $\sigma^2 = 1$ and use one-step GMM with identity weight to estimate the mean. The first uses
    $g_1(X_i,\theta) = (X_i - \theta,\, (X_i-\theta)^2 - 1)'$ as in Example~\ref{exm:schennach}; the second uses
    $g_2(X_i,\theta) = (X_i - \theta,\, (X_i-\theta)^4 - 3)'$, which uses the kurtosis of the normal distribution. Because $X_i$ is symmetric about zero, the population GMM objective is even in $\theta$ for both moment functions; under regularity conditions analogous to the $\sigma^2 \geq 1/2$ requirement of Example~\ref{exm:schennach}, the pseudo-true value is $\theta_0 = 0$ in both cases. Direct calculation gives sensitivity $\Lambda_1 = \Lambda_2 = (1, 0)$. In contrast, their informativeness differs: $\Delta_1 = 1$, while
    \[
        \Delta_2 \;=\; \frac{(1+12\delta\sigma^2)^2}{1 + 24\delta\sigma^2 + 240 \delta^2 \sigma^4} \;<\; 1 \quad \text{whenever } \delta \neq 0,
    \]
    where $\delta = 3(\sigma^4 - 1)$ measures the degree of misspecification (with $\delta = 0$ at $\sigma^2 = 1$). The influence function of the second researcher contains a centered cubic term in $X_i$, which is correlated with the first moment but not in $\mathcal L(\nu_i)$, so it lowers informativeness by Corollary~\ref{cor:miss} without altering local sensitivity. Equivalently, $\Delta_2 = V_1/V_2$, where $V_1$ is the variance of the moment-channel component and $V_2$ is the total asymptotic variance. Derivations are in Appendix~\ref{app:schennach2}.
\end{example}

\section{Applications}\label{sec:app}

We apply the misspecification-robust sensitivity and informativeness diagnostics to three canonical structural settings, each chosen to illustrate a different combination of the influence-function channels in Proposition~\ref{prp:gmm_ifs}. BLP combines a parameter-dependent efficient weight with a large Jacobian channel; BPP isolates the weight-matrix channel in a minimum-distance setting where the Jacobian is deterministic; and AJRY likewise has a parameter-dependent weight, but a smaller cluster-level Jacobian channel, and is the setting in which we examine the iteration path and the centering of the Hansen $J$ statistic. The BLP and AJRY sections examine how iterating the weight matrix reshapes the two diagnostics.

\subsection{BLP Automobile Market}
\label{sec:blp}

Our first application revisits the automobile demand and supply model of \citet*{berry1995automobile} (BLP), following the implementation in \citet*{andrews2017measuring} (AGS). The purpose of this application is to illustrate how allowing for misspecified moments alters sensitivity diagnostics relative to AGS, holding the model, data, and estimation procedure fixed.

The model is identical to that in AGS. Let $X_j$ denote the observable characteristics for product $j$ in a given market, $p_j$ its observed price, and $Z_j$ the vector of instruments. The unobservables $u_j(\theta) = (\xi_j(\theta),\eta_j(\theta))'$ collect the unobserved component of product quality and the unobserved component of marginal cost. Although $\xi_j(\theta)$ and $\eta_j(\theta)$ are not directly observed, they are recovered from the model at any candidate parameter $\theta$: $\xi_j(\theta)$ is obtained by inverting the market-share equation following \citet*{berry1994estimating}, and $\eta_j(\theta)$ is obtained from the firms' Bertrand-Nash pricing first-order conditions after recovering the model-implied marginal cost $mc_j(\theta)$. The moment function $g(X_j,\theta) = Z_j \otimes u_j(\theta)$ is therefore computable at each $\theta$, even though the structural unobservables themselves are not. Stacking the demand and supply moments yields a 31-dimensional sample moment vector $\widehat g(\theta) = n^{-1}\sum_{j=1}^n g(X_j,\theta)$.

Estimation is conducted by efficient two-step GMM. The second-step estimator $\widehat\theta$ minimizes
\[
    \widehat Q(\theta) = \widehat g(\theta)' \widehat W(\widehat\phi)\, \widehat g(\theta),
\]
where $\widehat W(\widehat\phi) = \bigl\{n^{-1}\sum_j g(X_j,\widehat\phi)g(X_j,\widehat\phi)'\bigr\}^{-1}$ is the estimated efficient weight matrix and $\widehat\phi$ is the first-step GMM estimator with identity weight matrix.

The parameter of interest is the average markup,
\[
    c(\theta) = \frac{1}{n}\sum_j \frac{p_j - mc_j(\theta)}{p_j},
\]
where $p_j$ is the observed price and $mc_j(\theta)$ is the model-implied marginal cost recovered from the supply-side first-order conditions. The influence function of $c(\theta)$ follows from the chain rule: $\partial c(\theta)/\partial\theta'$ multiplied by the influence function of the GMM estimator.

AGS study sensitivity under local violations of instrument validity, modeled as direct effects of instruments entering demand or cost equations at rate \(1/\sqrt n\). We compute the corresponding normalized sensitivity coefficients using both AGS and our MRS measure for the same excluded instruments. The AGS sensitivities are taken directly from their replication files, while MRS is computed using the misspecification-robust influence function derived in Proposition~\ref{prp:gmm_ifs}(iii).

Figure~\ref{fig:blpsensecoef} compares the two sensitivity measures. The attenuation under MRS is most pronounced for the demand-side instruments. In the AGS calculation, several demand-side instruments have sizable positive sensitivity for the average markup. Under MRS, these demand-side sensitivities are close to zero and in some cases change sign. On the supply side, the comparison is more heterogeneous: MRS often reduces sensitivity, but some supply-side moments remain influential and a number of rankings change. The main empirical finding is that allowing for misspecification substantially reorders the sensitivity profile. Table~\ref{tab:blp_bias} revisits the asymptotic bias calculations in AGS using MRS, showing that violations related to economies of scope and to cross-firm demand spillovers have substantially smaller effects under MRS relative to AGS.

The Hansen $J$ statistic is $345$ on $25$ degrees of freedom with $p<0.001$, decisively rejecting the moment conditions. By Corollary~\ref{cor:miss}, $\widehat\Delta_{markup}<1$ reflects the part of the influence function outside the moment span $\mathcal L(\nu_i)$. Both channels are active in BLP: the Jacobian channel $M_\gamma\gamma_i$ because $G(X_j,\theta)$ varies non-trivially across products through the nonlinear demand system, and the weight-matrix channel $M_\omega\omega_i$ because $\widehat W$ depends on the first-step $\widehat\phi$. Their joint contribution gives $\widehat\Delta_{markup}=0.56$:\footnote{The MRS-based quantities for BLP are computed in MATLAB R2026a; under R2024a the average-markup informativeness is $0.58$ in place of $0.56$, a second-decimal toolchain difference with no effect on any conclusion.} even after conditioning on the full vector of moments, $44\%$ of the asymptotic variance of the markup estimator is unexplained by the linear projection on the moments. The $J$-test detects misspecification; $\Delta$ quantifies the share of structural efficiency it costs.

For completeness, we also compute the iterated GMM analogue using Proposition~\ref{prp:gmm_ifs}(iv). The spectral radius $\rho(-A^{-1}B)\approx 0.64$ remains below one along the iteration path, so the iterated estimator is locally well-defined; the iteration converges within a small number of weight updates, leaving the Hansen $J$ essentially unchanged at $\approx 345$, while the iterated $\widehat\Delta_{markup}=0.24$ falls well below the two-step value of $0.56$.

\begin{figure}[H]
    \centering
    \includegraphics[width=0.95\linewidth]{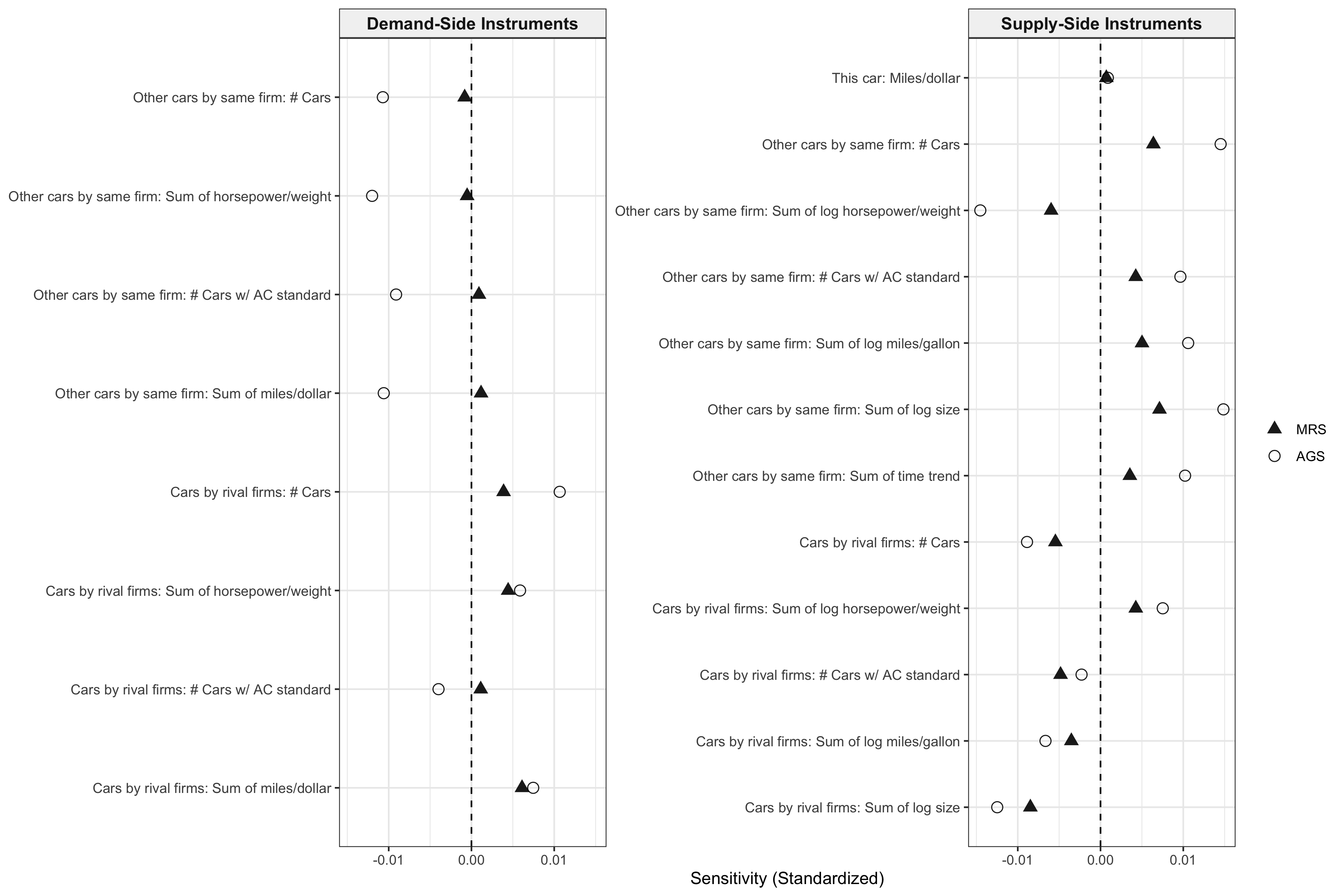}
    \caption{AGS and MRS sensitivities of the average markup with respect to the
    excluded instruments in the BLP automobile demand and supply model. Hollow
    circles plot the AGS sensitivities of \citet{andrews2017measuring}; filled
    triangles plot the misspecification-robust sensitivities derived from
    Proposition~\ref{prp:gmm_ifs}(iii). Sensitivities are reported as the
    asymptotic effect of a one-standard-deviation perturbation of each
    instrument, expressed in units of one percent of the average price.}
    \label{fig:blpsensecoef}
\end{figure}

\begin{table}[ht]
    \centering
    \begin{tabular}{lcc}
        \toprule
        Misspecification scenario            & AGS bias (pp) & MRS bias (pp) \\
        \midrule
        Supply moments, same-firm spillover  & $-17.3$       & $-7.1$        \\
        Supply moments, rival-firm spillover & $+20.9$       & $+16.0$       \\
        Demand moments, same-firm spillover  & $-12.8$       & $-1.3$        \\
        Demand moments, rival-firm spillover & $+25.2$       & $+9.0$        \\
        \bottomrule
    \end{tabular}
    \caption{Asymptotic bias of the average markup under four local violations of an excluded-instrument restriction at rate $1/\sqrt n$, normalized to one percent of the average price. ``AGS bias'' uses the AGS sensitivity; ``MRS bias'' uses Proposition~\ref{prp:gmm_ifs}(iii). Units: percentage points of the average markup.}
    \label{tab:blp_bias}
\end{table}

\subsection{BPP Income Dynamics and Consumption Insurance}
\label{sec:bpp}

Our second application revisits the model of household income dynamics and consumption insurance of \citet*{blundell2008consumption} (BPP). BPP estimate the model by diagonally weighted minimum distance (DWMD); we report both that scheme and the optimal weight matrix, and the contrast between the two is itself diagnostic.

The BPP model also provides a clean illustration of Proposition~\ref{prp:eff_channel} and Corollary~\ref{cor:miss}. By construction, the Jacobian of the MD estimator is deterministic across households, so $\gamma_i = 0$ and the Jacobian channel vanishes. Under the optimal weight the total misspecification residual reduces to the weight-matrix channel $r_k(X_i) = M_\omega\omega_i$, the rank-one factorization of Proposition~\ref{prp:eff_channel}, which generates structural-efficiency loss whenever this term has a nonzero component orthogonal to $\mathcal L(\nu_i)$. We estimate the BPP model with time-varying shock variances, comprising $35$ parameters and $325$ second-moment conditions. Even when the Jacobian contributes no sampling variation, the weight-matrix channel can drive informativeness substantially below one.

The BPP model decomposes log income into a permanent and a transitory component, with consumption growth depending on permanent and transitory income shocks through partial-insurance coefficients $\phi$ and $\psi$. The full structural equations and the construction of the non-collinear moment set are given in Appendix~\ref{sec:bppmoment}. We estimate the specification with time-varying variances of the permanent shock, the transitory shock, and the consumption innovation, holding $\phi$ and $\psi$ constant across periods; the parameter vector $\beta$ collects these time-varying variances along with $\phi$, $\psi$, the MA coefficient, and the measurement-error variance. Let $m_i$ denote the household-level vector of sample cross-products of income and consumption growth,\footnote{The i.i.d. assumption is imposed at the household level. Although the moments use within-household time-series covariances, these are collected into the household-level vector $m_i$, and within-household serial dependence is absorbed into $m_i$. The asymptotic approximation treats households as independently sampled.} and let $m(\beta)$ denote the corresponding theoretical covariances. The sample moment conditions are $g(X_i, \beta) = m_i - m(\beta)$.

 The OMD estimator minimizes $\widehat g(\beta)'\widehat W\widehat g(\beta)$ with $\widehat W = [n^{-1}\sum_i (m_i - \bar m)(m_i - \bar m)']^{-1}$ where $\bar{m}=n^{-1}\sum_{i}m_{i}$,  the inverse of the sample covariance matrix of the data moments.\footnote{Our construction of this sample covariance follows \citet*[Appendix D]{blundell2008consumption}, except that for unbalanced panels each entry is divided by the product of the two moments' observation counts rather than by the number of households for whom both moments are observed. The two coincide for a balanced panel. For the cohort-A subsample BPP's normalization produces a covariance estimate that is not positive-definite and so cannot be used to form the OMD or DWMD weight matrix; ours fixes this with no other change to the construction.} Because $\widehat W$ is built from the data moments alone, it does not depend on $\beta$: OMD is one-step GMM with an estimated weight matrix, its influence function is given by Proposition~\ref{prp:eff_channel}, and there is nothing to iterate.\footnote{In the unbalanced panel the per-moment subsamples leave a residual dependence of the weight on $\beta$ through the missing-data pattern, so the formally iterated estimator is not exactly the one-step; re-estimating with the updated weight converges in three steps with negligible movement in the estimates.} 
 
 We also estimate the model by DWMD, the scheme used by BPP, with weight $\widehat D = (\operatorname{diag}\widehat V)^{-1}$, weighting each moment by the inverse of its own sampling variance, and by EWMD, the equal-weighting scheme of \citet*{altonji1996small}, with the identity weight. The three estimators share the same moments and model map and differ only in how much of the weight is estimated: all of $\widehat V^{-1}$ for OMD, its diagonal for DWMD, none for EWMD.

Parameter estimates and standard errors are reported in Table~\ref{tab:bppparam}. The Hansen $J$ statistic for overidentifying restrictions is $538.7$ on $290$ degrees of freedom with $p<0.001$, indicating that the BPP moment conditions are decisively rejected.

Table~\ref{tab:bppparam} also reports the misspecification-robust informativeness. Under OMD, $\widehat\Delta = 0.79$ for both $\phi$ and $\psi$; across all $35$ parameters $\widehat\Delta_k$ ranges from $0.70$ to $0.87$ with median $0.79$. The misspecification-robust standard errors exceed the conventional ones ($0.043$ versus $0.034$ for $\phi$); the gap is exactly the weight-matrix channel that the conventional formula omits.

The results under DWMD are markedly different. The partial-insurance coefficients are estimated at $\widehat\phi=0.68$ and $\widehat\psi=0.03$, close to the values reported by BPP and well away from the OMD estimates of $0.33$ and $0.07$; EWMD gives $\widehat\phi=0.61$ and $\widehat\psi=0.00$. Since the three weightings share a probability limit under correct specification, this divergence is itself a symptom of misspecification. The DWMD informativeness, however, is $\widehat\Delta\approx0.99$ for $\phi$, for $\psi$, and for every one of the $35$ parameters. Figure~\ref{fig:BPPdelta} plots the full distribution under both weightings: the OMD values are scattered between $0.70$ and $0.87$, while the DWMD values cluster just below one.

The contrast is structural. OMD's weight inverts the sample covariance of the moments. By Proposition~\ref{prp:eff_channel}, once the overidentifying restrictions are rejected, this construction makes the estimated weight contribute variation that the moments cannot account for, and the contribution grows with the strength of the rejection. With the BPP moments decisively rejected, roughly a fifth of the asymptotic variance of $\widehat\phi$ and $\widehat\psi$ reflects estimation of the weight matrix rather than the moments. DWMD keeps only the diagonal of that weight and discards the cross-moment terms, so it carries no comparable contribution; its variance stays almost entirely explained by the moments, and $\widehat\Delta\approx0.99$. EWMD is the limiting case: the identity weight estimates nothing, so by Proposition~\ref{prp:eff_channel} the weight-matrix channel is absent and $\Delta_k=1$ for every parameter, which the numerical computation reproduces exactly.

The simpler weights are not free: the misspecification-robust standard error of $\phi$ rises from $0.043$ under OMD to $0.113$ under DWMD and $0.109$ under EWMD, the efficiency-informativeness trade of Remark~\ref{rem:altonji}. The loss of informativeness in BPP is thus a feature of the optimal weight matrix specifically, not of weighted minimum distance in general.
\begin{table}[ht]
    \centering
    \begin{tabular}{lcccccccc}
        \toprule
                      & \multicolumn{4}{c}{$\phi$ (permanent)} & \multicolumn{4}{c}{$\psi$ (transitory)}                              \\
        \cmidrule(lr){2-5}\cmidrule(lr){6-9}
        Estimator     & Est.                                   & SE (conv.)                              & SE (MR) & $\widehat\Delta$
                      & Est.                                   & SE (conv.)                              & SE (MR) & $\widehat\Delta$ \\
        \midrule
        OMD           & $0.327$                                & $0.034$                                 & $0.043$ & $0.79$
                      & $0.070$                                & $0.029$                                 & $0.039$ & $0.79$           \\
        DWMD          & $0.682$                                & $0.093$                                 & $0.113$ & $0.99$
                      & $0.028$                                & $0.043$                                 & $0.043$ & $0.99$           \\
        EWMD          & $0.610$                                & $0.108$                                 & $0.109$ & $1.00$
                      & $0.003$                                & $0.047$                                 & $0.049$ & $1.00$           \\
        \bottomrule
    \end{tabular}
    \caption{Partial-insurance coefficient estimates for the full BPP model ($35$ parameters, $325$ moment conditions, $n=1765$ households). $\phi$ and $\psi$ are consumption responses to permanent and transitory shocks; ``SE (conv.)'' omits the weight-matrix channel, ``SE (MR)'' is misspecification-robust, $\widehat\Delta$ is misspecification-robust informativeness. OMD uses the optimal weight $\widehat V^{-1}$ (Proposition~\ref{prp:eff_channel}), DWMD the diagonal weight $(\operatorname{diag}\widehat V)^{-1}$, EWMD the identity weight.}
    \label{tab:bppparam}
\end{table}

\begin{figure}[ht]
    \centering
    \includegraphics[width=0.7\linewidth]{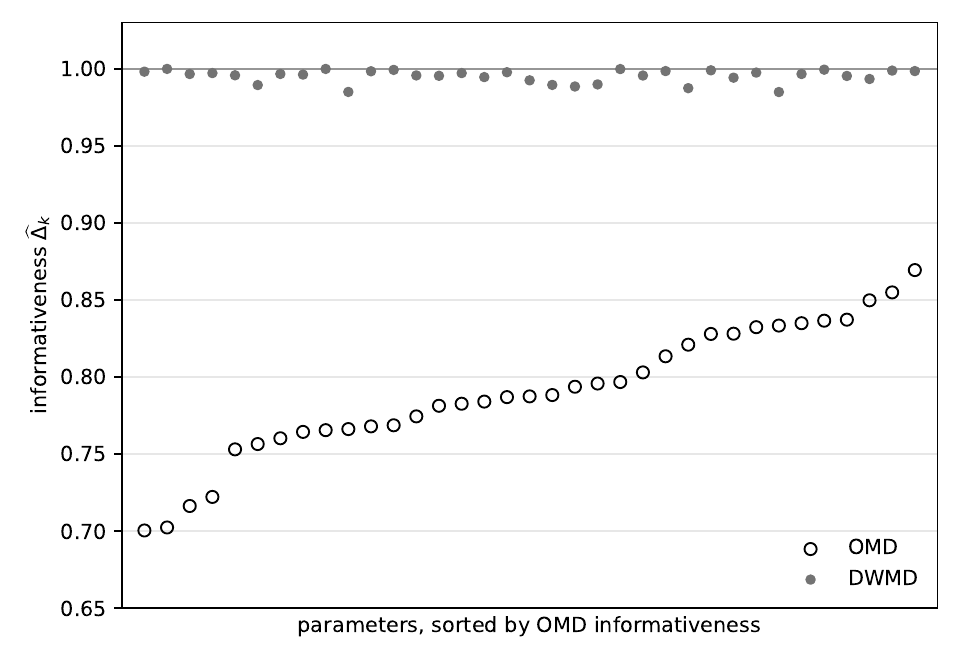}
    \caption{Misspecification-robust informativeness $\widehat\Delta_k$ of every BPP parameter under OMD (hollow circles) and DWMD (filled circles), sorted by the OMD value. The thin line marks the fixed-weight bound $\Delta_k=1$, attained exactly by EWMD.}
    \label{fig:BPPdelta}
\end{figure}

\subsection{Income and Democracy}
\label{sec:ajry}

Our third application revisits the income-and-democracy regressions of \citet*{acemoglu2008income} (AJRY), a leading example of dynamic panel difference-GMM. AJRY estimate the autoregressive specification
\[
    d_{it} = \alpha\,d_{i,t-1} + \gamma\,y_{i,t-1} + \mu_t + \delta_i + u_{it},
\]
where $d_{it}$ is a measure of democracy, $y_{i,t-1}$ is lagged log income per capita, $\mu_t$ are time effects, and $\delta_i$ are country fixed effects. First-differencing eliminates $\delta_i$, and the resulting moment conditions $\mathbb E[Z_i\Delta u_{i}] = 0$ use Arellano-Bond instruments built from lagged levels of $d$ and $y$. Because the data form a country panel, the cross-sectional unit is the country and within-country serial dependence is absorbed into the stacked cluster moment $Z_i'\Delta u_i$. The influence functions and standard errors are cluster-robust, computed at the country level as in Appendix~\ref{app:cluster}. The parameter of interest is $\gamma$, the effect of income on democracy.

We replicate the five-year panel specification of \citet*{hansen2021inference}, which reports the iterated GMM estimator of AJRY with misspecification-robust standard errors. The iterated estimate is $\widehat\alpha = 0.744$ and $\widehat\gamma = -0.009$, with misspecification-robust standard errors $0.128$ and $0.039$ respectively. The Hansen $J$ statistic at the iterated fixed point is $45.3$ on $44$ degrees of freedom ($p=0.42$), failing to reject at conventional levels.\footnote{\label{fn:ajryJ}The Hansen $J$ statistic is
$J = n\,\widehat g'\widehat W\widehat g$, with $\widehat g$ the mean cluster
moment vector and $\widehat W$ the inverse of the uncentered cluster
second-moment matrix $n^{-1}\sum_i(Z_i'\Delta \widehat{u}_i)(Z_i'\Delta \widehat{u}_i)'$,
evaluated at the relevant parameter estimate.} The one-step estimator that AJRY report uses the deterministic Arellano-Bond weight matrix and yields a Hansen $J$ of $71.3$ ($p=0.006$). The two-step efficient estimator yields $\widehat\gamma = -0.012$ with a robust standard error of $0.049$ and a Hansen $J$ of $61.4$ ($p=0.04$). All three estimators find the income coefficient statistically insignificant, but they deliver materially different specification-test verdicts.

AJRY report a one-step estimate of $\gamma$ that is statistically insignificant, concluding that income has no causal effect on democracy. Its misspecification-robust informativeness is $\widehat\Delta_\gamma = 0.93$. The Arellano-Bond weight does not depend on the parameter, so the one-step estimator activates the Jacobian channel but no weight-matrix channel; the modest loss is the Jacobian channel alone. The Jacobian channel is active because the per-unit Jacobian $G(X_i,\theta)$ varies across countries through the lagged democracy and income terms in $Z_i'\Delta X_i$. Adopting the efficient weight matrix lowers informativeness rather than raising it. The two-step efficient estimator has $\widehat\Delta_\gamma = 0.81$ and the iterated estimator $\widehat\Delta_\gamma = 0.77$. Figure~\ref{fig:AJRYjdelta} plots the $J$ statistic and $\widehat\Delta_\gamma$ along the iteration path, and the two diagnostics give opposite verdicts on iteration. Iterating the weight matrix moves the overidentifying-restrictions test from rejection at the one-step and two-step estimators ($p=0.006$ and $p=0.04$) to non-rejection at the fixed point ($p=0.42$), while $\widehat\Delta_\gamma$ falls in parallel from $0.93$ to $0.77$.

The non-rejection, however, depends on how the weight matrix is centered. The statistic of footnote~\ref{fn:ajryJ} uses the uncentered second-moment matrix; with the centered covariance, the two statistics at the same parameter value are linked by $J = J_c/(1+J_c/n)$, by the Sherman--Morrison formula. Under correct specification the two are asymptotically equivalent and either yields a valid $\chi^2$ test, so the centering is usually treated as innocuous. Under misspecification they diverge by a first-order amount, and the centered version is the more powerful test \citep*{hall2000covariance}. Figure~\ref{fig:AJRYjdelta} plots the two statistics along their respective iteration paths. The centered and uncentered iterations converge to the same fixed-point estimator, where the uncentered statistic is below the critical value ($45.3$, $p=0.42$) and the centered statistic far above it ($70.4$, $p=0.007$); the centered test rejects at every step of its path. A specification verdict that flips with the centering convention is not a reliable reading of misspecification. The informativeness, by contrast, is invariant to the convention: the centered and uncentered weights produce the same iterated estimator and influence function, so $\widehat\Delta_\gamma=0.77$ either way, and the structural-efficiency loss is a property of the estimator rather than of the test.

\begin{figure}[ht]
    \centering
    \includegraphics[width=0.8\linewidth]{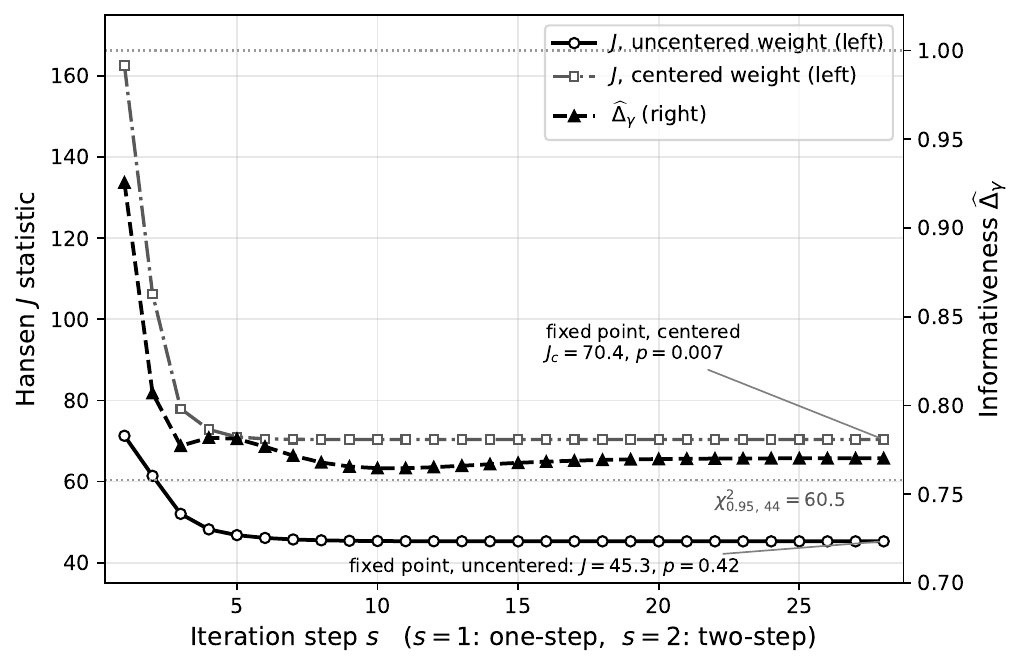}
    \caption{Hansen $J$ statistics and informativeness $\widehat\Delta_\gamma$ along the iteration paths in the \citet*{acemoglu2008income} specification (five-year panel, full sample). Circles: the uncentered efficient-weight iteration; squares: the fully centered iteration, with the weight equal to the inverse cluster covariance updated at each step. Both start from the Arellano-Bond one-step estimate ($s=1$) and converge to the same fixed-point estimator, where $J = J_c/(1+J_c/n)$ with $n=127$ countries: $J=45.3$ versus $J_c=70.4$. Dotted line: the $5\%$ critical value $\chi^2_{0.95,44}=60.5$. $\widehat\Delta_\gamma$ (right axis) is shown along the uncentered path.}
    \label{fig:AJRYjdelta}
\end{figure}

Figure~\ref{fig:AJRYmulti} traces the efficient-weight iteration from five first-step weights. Panel (a) reproduces Figure~1(b) of \citet*{hansen2021inference}: the income coefficient converges to a common fixed point from every start, consistent with the contraction rate $\rho(-A^{-1}B)\approx 0.61$ below one. The informativeness likewise reaches the common fixed-point value $\widehat\Delta_{\gamma}=0.77$ from every start (panel b). Because the $s$-step estimators along the path converge to distinct pseudo-true values, their step-wise informativeness values are not directly comparable, but the panel shows that the approach to the fixed point is non-monotone.

\begin{figure}[ht]
    \centering
    \includegraphics[width=\linewidth]{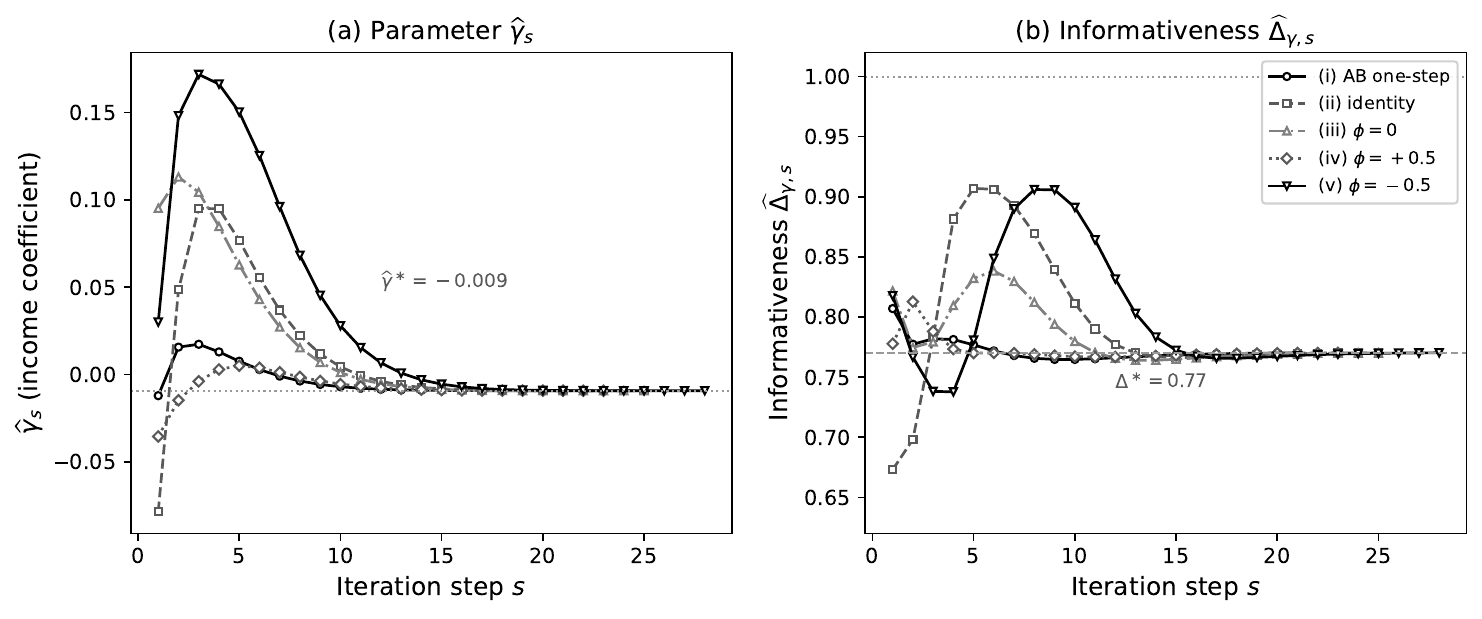}
    \caption{Iteration paths of difference GMM in AJRY specification (five-year panel, full sample) from five first-step weights
        $W(\phi)$, with $\phi$ equal to the Arellano-Bond one-step estimate, the
        identity-weight one-step estimate, the zero vector, 0.5 times the vector of ones, and
        -0.5 times the vector of ones. Panel (a) reproduces Figure~1(b) of
        \citet*{hansen2021inference}: the income coefficient $\widehat\gamma_s$
        converges to a common fixed point $\widehat\gamma^\ast=-0.009$ from every
        start. Panel (b):
        the informativeness $\widehat\Delta_{\gamma,s}$ evaluated at each step's
        estimator. The iteration uses the uncentered efficient weight. Both
        panels use the recursive $s$-step
        influence function; the iterated influence function of
        Proposition~\ref{prp:gmm_ifs}(iv) is its $s\to\infty$ limit.}
    \label{fig:AJRYmulti}
\end{figure}

Iteration affects the two diagnostics differently across the applications. In AJRY both move, the test passing to non-rejection while $\widehat\Delta_\gamma$ falls from $0.81$ to $0.77$. In BLP, by contrast, iteration leaves the Hansen $J$ essentially unchanged at $\approx 345$ while $\widehat\Delta_{markup}$ falls sharply from $0.56$ to $0.24$ (Section~\ref{sec:blp}). $J$ and $\widehat\Delta_k$ measure different functionals of the same misspecification: $J$ records how far the moments depart from zero at the pseudo-true value, and $\widehat\Delta_k$ the share of the estimator's variance the moments explain, so iteration reshapes the Jacobian and weight-matrix channels, large in BLP and modest in AJRY, without moving $J$ in step.

\section{Conclusion}\label{sec:conclusion}

This paper develops a sensitivity and informativeness framework for GMM estimators that remains valid under general moment misspecification. By expressing sensitivity measures through influence functions, we provide a unified characterization of how deviations in moment conditions affect GMM estimators across one-step, two-step, and iterated procedures, with continuously updating GMM treated in Appendix~\ref{sec:cugmm}. Once the misspecification-robust influence function is in hand, the sensitivity matrix $\Lambda$ and the informativeness measure $\Delta_{k}$ follow at essentially no additional cost.

The informativeness measure $\Delta_{k}$ is the diagnostic contribution. It quantifies the share of an estimator's asymptotic variance driven by sampling variation in the moments, as opposed to variation arising from the Jacobian or the estimated weight matrix, and is complementary to the Hansen $J$-test: the $J$-test asks whether the moment restrictions are jointly violated; $\Delta_{k}$ asks how much of the estimator's variance is left unexplained by the moments after such a violation. The income-and-democracy application illustrates the complementarity. At the iterated fixed point the $J$-test verdict depends on the centering of the weight matrix ($p=0.42$ uncentered, $p=0.007$ centered), while $\widehat\Delta_\gamma = 0.77$ under either convention, so nearly a quarter of the income coefficient's asymptotic variance is unexplained by the moments regardless of how the test is run.

A practical implication concerns the choice of weight matrix. In minimum distance estimation, the optimal weight that delivers classical efficiency under correct specification carries a hidden cost under misspecification: a loss of informativeness, or structural efficiency, that is absent under correct specification. By Proposition~\ref{prp:eff_channel}, the very construction that attains this efficiency, the inversion of the moments' second-moment matrix, is the one that makes the estimated weight contribute variation the moments cannot account for once the overidentifying restrictions are rejected, and this informativeness loss grows with the strength of the rejection. Simpler weights avoid it. Diagonally weighted minimum distance, which inverts only the moments' marginal variances, raises $\widehat\Delta$ for all $35$ parameters of our BPP application from a median of $0.79$ to about $0.99$, and equal weighting, which estimates no weight at all, attains $\Delta_{k}=1$ exactly. Pursuing classical efficiency under misspecification therefore trades informativeness for an efficiency gain the optimal weight no longer delivers.

We recommend that practitioners using over-identified GMM report $\Delta_{k}$ alongside misspecification-robust standard errors. Where $\Delta_{k}$ is well below one, the moments explain only a small share of the estimator's asymptotic variance, and any economic interpretation of the parameter should be made with that in mind. The framework does not correct for misspecification; it clarifies how and through which channels misspecification enters the estimator.

%\section*{Declaration of generative AI and AI-assisted technologies in the manuscript preparation process}
%\vspace{-1em}
%During the preparation of this work the authors used Claude (Anthropic) in order to check the code and to assist with editing and improving the language of the manuscript. After using this tool, the authors reviewed and edited the content as needed and take full responsibility for the content of the published article.

%%%%%%%%%%%%%%%%%%%%%%%%%%%%%%%%%%%%%%%%%%%%%%%%
%\begin{singlespace}
\bibliographystyle{ecta}
\bibliography{references.bib}
%\end{singlespace}
%%%%%%%%%%%%%%%%%%%%%%%%%%%%%%%%%%%%%%%%%%%%%%%%%

%%%%%%%%%%%%%%%%%%%%%%%%%%%%%%%%%%%%%%%%%%%%%%%%%
%%%%% These commands start the appendix and change the Table & Figure numbering
%\newpage
\appendix
\setcounter{table}{0}
\renewcommand{\tablename}{Appendix Table}
\renewcommand{\figurename}{Appendix Figure}
\renewcommand{\thetable}{A\arabic{table}}
\setcounter{figure}{0}
\renewcommand{\thefigure}{A\arabic{figure}}
\setcounter{equation}{0}
\renewcommand{\theequation}{A\arabic{equation}}
%%%%%%%%%%%%%%%%%%%%%%%%%%%%%%%%%%%%%%%%%%%%%%%%%

\section{Appendix: Proofs \label{sec:proofs}}

\begin{proof}[Proof of Proposition \ref{prp:gmm_ifs}]
    We begin by establishing preliminary uniform convergence results implied by
    Assumption~\ref{ass:regularity}.

    Assumptions~\ref{ass:regularity}(i)--(iii) imply a Uniform Law of Large Numbers (ULLN) over the parameter space \(\Theta\) for the sample moments and their derivatives that appear in the sample first-order conditions \citep*[see, e.g.,][Theorem 3]{andrews1992generic}. In particular,
    \begin{align*}
        \sup_{\theta\in\Theta}\|\widehat g(\theta)-\mathbb E[g(X_i,\theta)]\|
         & \xrightarrow{p}0, \\
        \sup_{\theta\in\Theta}\|\widehat G(\theta)-\mathbb E[G(X_i,\theta)]\|
         & \xrightarrow{p}0, \\
        \sup_{\theta\in\Theta}\|\widehat R(\theta)-\mathbb E[R(X_i,\theta)]\|
         & \xrightarrow{p}0.
    \end{align*}
    Since \(H(\theta)\) is an affine function of \(R(\theta)\) with coefficients
    depending on \(g(\theta)\) and the weight matrix \(W\), it follows that
    \[
        \sup_{\theta\in\Theta}\|\widehat H(\theta)-\mathbb E[H(X_i,\theta)]\|
        \xrightarrow{p}0.
    \]

    For one-step and two-step GMM, uniform convergence of the sample objective function and its derivatives follows from Assumption~\ref{ass:regularity}(i)--(iii) and (vi). In particular, the sample criterion $Q_n(\theta)$ converges uniformly to its population counterpart $Q(\theta)$. Standard extremum estimator arguments \citep*[e.g.,][Theorem~2.1]{newey1994large} therefore imply consistency, $\widehat\theta\xrightarrow{p}\theta_0$ (and $\widehat\phi\xrightarrow{p}\phi_0$ for the first step).

    For iterated GMM, the contraction condition in Assumption~\ref{ass:regularity}(v), together with the ULLN over $\Theta$, ensures that the sample updating map converges uniformly to its population counterpart and that the iterated GMM sequence converges with probability approaching one to the unique fixed point $\theta_0$. The argument is directly analogous to that of \citet*{hansen2021inference}, who establish existence and consistency of the iterated GMM estimator under misspecification.

    Since \(\theta_0\) is the unique interior minimizer of \(Q(\theta)\), the population
    first-order condition holds
    \[
        G'Wg=0,
    \]
    where \(W\) is the corresponding limit weight matrix.

    With probability approaching one, the estimator satisfies the sample
    first-order condition (FOC) \(F_n(\widehat\theta)=0\), where \(F_n(\theta) = \widehat G(\theta)'\widehat W\widehat g(\theta)\) is the \(p\times 1\) gradient of the GMM criterion and \(\widehat W\) is the weight matrix used by the estimator: deterministic for one-step GMM, \(\widehat W(\widehat\phi)\) for two-step GMM, and \(\widehat W(\theta)\) for iterated GMM. We write \(F_{n,\theta}(\theta) = \partial F_n(\theta)/\partial\theta'\) for its Jacobian, and \(F_{n,\phi}(\theta,\phi) = \partial F_n(\theta,\phi)/\partial\phi'\) when the weight depends on a first-step parameter \(\phi\). The explicit forms of these derivatives for each estimator are given in the corresponding cases below. Because \(F_n(\theta)\) is a \(p \times 1\) vector, we apply the mean value theorem row-by-row around \(\theta_0\), which yields
    \[
        0 = F_n(\theta_0)
        + F_{n,\theta}(\widetilde\theta)(\widehat\theta-\theta_0),
    \]
    where \(F_{n,\theta}(\widetilde\theta)\) is a matrix whose \(k\)-th row is evaluated at an intermediate point \(\widetilde\theta^{(k)}\) on the line segment between \(\widehat\theta\) and \(\theta_0\). Because \(\widehat\theta\xrightarrow{p}\theta_0\), every \(\widetilde\theta^{(k)}\xrightarrow{p}\theta_0\), and by the ULLN, the sample Jacobian \(F_{n,\theta}(\widetilde\theta)\) converges uniformly in probability to the appropriate population curvature matrix.

    Rearranging,
    \[
        \sqrt n(\widehat\theta-\theta_0)
        = -\bigl[F_{n,\theta}(\widetilde\theta)\bigr]^{-1}\sqrt n F_n(\theta_0).
    \]
    Thus, for each estimator, the derivation proceeds by (i) establishing the probability limit of
    \(F_{n,\theta}(\widetilde\theta)\), and (ii) obtaining an asymptotic linear expansion
    \[
        \sqrt n F_n(\theta_0)
        = \frac1{\sqrt n}\sum_{i=1}^n f(X_i)+o_p(1).
    \]
    The influence function then follows as \(\psi(X_i)=-A^{-1}f(X_i)\).

    For two-step GMM, additional moment conditions are
    required to derive asymptotic linear representations. In particular,
    Assumption~\ref{ass:regularity}(vi)(b) guarantees the finite fourth moments
    needed to apply a central limit theorem to
    \(g(X_i,\theta_0)g(X_i,\theta_0)'\) and hence to justify the influence function
    expansions for efficient GMM estimators.

    To formalize the \(o_p(1)\) remainder terms in the following multivariable Taylor expansions, let \(\Delta G = \widehat{G}(\theta_0) - G\), \(\Delta g = \widehat{g}(\theta_0) - g\), and \(\Delta W = \widehat{W} - W\) (evaluated at the relevant probability limit). By the Central Limit Theorem and Assumption~\ref{ass:regularity}, \(\Delta G\), \(\Delta g\), and \(\Delta W\) are all \(O_p(n^{-1/2})\). Consequently, any pairwise product (e.g., \(\Delta G' W \Delta g\)) is \(O_p(n^{-1})\), and any three-way product (e.g., \(\Delta G' \Delta W \Delta g\)) is \(O_p(n^{-3/2})\). When scaled by \(\sqrt{n}\), these higher-order cross-products are at most \(O_p(n^{-1/2})\) and are therefore \(o_p(1)\).

    \vspace{0.5em}
    \noindent\textbf{One-Step GMM with Deterministic Weight $W$}

    The sample FOC is $F_n(\theta) = \widehat G(\theta)' W \widehat g(\theta)$ with derivative
    \begin{equation}
        \label{eq:one_step_foc}
        F_{n,\theta}(\theta) = \widehat G(\theta)' W \widehat G(\theta) + (\widehat g(\theta)' W \otimes I_p)\widehat R(\theta).
    \end{equation}
    By the ULLN and Continuous Mapping Theorem,
    \[
        F_{n,\theta}(\widetilde\theta) \xrightarrow{p} G'WG + H = A.
    \]
    Expanding $F_n(\theta_0)$ around the population limits $G$ and $g$ and using $G'Wg = 0$,
    \[
        \sqrt{n}F_n(\theta_0) = G'W\sqrt{n}(\widehat g - g) + \sqrt{n}(\widehat G - G)'Wg + o_p(1).
    \]
    The influence function of the FOC is, by the vectorization identity $(G(X_i,\theta_0) - G)'Wg = (g'W\otimes I_p)\gamma_i$,
    \[
        f_1(X_i) = G'W\nu_i + (g'W\otimes I_p)\gamma_i.
    \]
    Hence $\psi^{1s}(X_i) = -A^{-1}f_1(X_i)$, matching Proposition~\ref{prp:gmm_ifs}(i).

    \vspace{0.5em}
    \noindent\textbf{One-Step GMM with $\widehat W$ Estimated Independent of $\theta$}

    The FOC is $F_n(\theta) = \widehat G(\theta)' \widehat W \widehat g(\theta)$. Since $\widehat W$ does not depend on $\theta$ and $\widehat W \xrightarrow{p} W$ by Assumption~\ref{ass:regularity}(vi)(a), $F_{n,\theta}(\theta)$ has the same form as \eqref{eq:one_step_foc} with $W$ replaced by $\widehat W$, and the limit remains $A$.

    Expanding $\widehat G$, $\widehat g$, and $\widehat W$ around their limits,
    \[
        \sqrt{n}F_n(\theta_0) = G'W\sqrt{n}(\widehat g - g) + \sqrt{n}(\widehat G - G)'Wg + G'\sqrt{n}(\widehat W - W)g + \mathcal R_n,
    \]
    where $\mathcal R_n$ collects higher-order cross-products and is $o_p(1)$. By the vectorization identity, the third term equals $(g' \otimes G')\sqrt{n}\,\vect(\widehat W - W)$, whose influence function is $(g'\otimes G')\omega_i$. Define
    \[
        f_2(X_i) = f_1(X_i) + (g'\otimes G')\omega_i.
    \]
    Then
    \[
        \psi^{1s,W}(X_i) = -A^{-1}f_2(X_i),
    \]
    matching Proposition~\ref{prp:gmm_ifs}(ii).

    \vspace{0.5em}
    \noindent\textbf{Two-Step GMM}

    Let $\widehat\phi$ be the first-step estimator, converging in probability to $\phi_0$ with influence function $\psi^\phi(X_i)$. The second-step estimator $\widehat\theta$ uses $\widehat W(\widehat\phi)$ and converges to $\theta_0$, the minimizer of $Q(\theta) = g(\theta)'W(\phi_0)g(\theta)$. Generally $\phi_0 \neq \theta_0$ under misspecification.

    The second-step estimator satisfies $F_n(\widehat\theta, \widehat\phi) = \widehat G(\widehat\theta)'\widehat W(\widehat\phi)\widehat g(\widehat\theta) = 0$. Applying the mean value theorem row-by-row around $(\theta_0, \phi_0)$,
    \[
        0 = F_n(\theta_0, \phi_0) + F_{n,\theta}(\widetilde\theta, \widetilde\phi)(\widehat\theta - \theta_0) + F_{n,\phi}(\widetilde\theta, \widetilde\phi)(\widehat\phi - \phi_0),
    \]
    where $(\widetilde\theta, \widetilde\phi) \xrightarrow{p} (\theta_0, \phi_0)$. Rearranging,
    \[
        \sqrt{n}(\widehat\theta - \theta_0) = -\bigl[F_{n,\theta}(\widetilde\theta, \widetilde\phi)\bigr]^{-1}\Bigl(\sqrt{n}F_n(\theta_0, \phi_0) + F_{n,\phi}(\widetilde\theta, \widetilde\phi)\sqrt{n}(\widehat\phi - \phi_0)\Bigr).
    \]

    By the ULLN and CMT, $F_{n,\theta}(\widetilde\theta, \widetilde\phi) \xrightarrow{p} A = G'W(\phi_0)G + H$ with $H = (g'W(\phi_0)\otimes I_p)R$. The derivative with respect to $\phi$ is $F_{n,\phi}(\theta, \phi) = (\widehat g(\theta)'\otimes \widehat G(\theta)')\,\partial\vect(\widehat W(\phi))/\partial\phi'$. By Assumption~\ref{ass:regularity}(vi)(b), $\partial\vect(\widehat W(\widetilde\phi))/\partial\phi' \xrightarrow{p} S = \partial\vect(W(\phi))/\partial\phi'|_{\phi_0}$, so $F_{n,\phi}(\widetilde\theta, \widetilde\phi) \xrightarrow{p} B = (g'\otimes G')S$.

    Expanding $\sqrt{n}F_n(\theta_0, \phi_0)$ around $(G, W(\phi_0), g)$ and using the FOC $G'W(\phi_0)g = 0$,
    \begin{align*}
        \sqrt{n}F_n(\theta_0, \phi_0) & = G'W(\phi_0)\sqrt{n}(\widehat g - g) + \sqrt{n}(\widehat G - G)'W(\phi_0)g \\
                                      & \quad + G'\sqrt{n}(\widehat W(\phi_0) - W(\phi_0))g + o_p(1).
    \end{align*}
    The first two terms produce $f_1(X_i)$ evaluated at $W(\phi_0)$, and the third produces $(g'\otimes G')\omega_i$. Together they recover $f_2(X_i)$ evaluated at $W(\phi_0)$. By Slutsky's theorem,
    \[
        \sqrt{n}(\widehat\theta - \theta_0) = -A^{-1}\!\left(\frac{1}{\sqrt{n}}\sum_{i=1}^n f_2(X_i) + B\,\frac{1}{\sqrt{n}}\sum_{i=1}^n \psi^\phi(X_i)\right) + o_p(1),
    \]
    so
    \[
        \psi^{2s}(X_i) = -A^{-1}\bigl(f_2(X_i) + B\psi^\phi(X_i)\bigr),
    \]
    matching Proposition~\ref{prp:gmm_ifs}(iii), with $A$, $B$, $f_2$ evaluated at $W(\phi_0)$.

    \vspace{0.5em}
    \noindent\textbf{Iterated GMM}

    At the fixed point, $\widehat\theta$ satisfies $F_n(\widehat\theta) = \widehat G(\widehat\theta)'\widehat W(\widehat\theta)\widehat g(\widehat\theta) = 0$. Applying the mean value theorem around $\theta_0$,
    \[
        0 = F_n(\theta_0) + F_{n,\theta}(\widetilde\theta)(\widehat\theta - \theta_0).
    \]
    The total derivative of $F_n(\theta)$ with respect to $\theta$ via the product rule is
    \[
        F_{n,\theta}(\theta) = \bigl[\widehat G(\theta)'\widehat W(\theta)\widehat G(\theta) + (\widehat g(\theta)'\widehat W(\theta)\otimes I_p)\widehat R(\theta)\bigr] + \bigl[(\widehat g(\theta)'\otimes \widehat G(\theta)')\,\partial\vect(\widehat W(\theta))/\partial\theta'\bigr].
    \]
    At the fixed point, $W = W(\theta_0)$. By the ULLN and Assumption~\ref{ass:regularity}(vi)(b), the first bracket converges to $A$ and the second to $B$, with $S = \partial\vect(W(\theta))/\partial\theta'|_{\theta_0}$ and $B = (g'\otimes G')S$. Hence $F_{n,\theta}(\widetilde\theta) \xrightarrow{p} A + B$.

    The expansion of $\sqrt{n}F_n(\theta_0)$ follows the two-step derivation evaluated at $W(\theta_0)$, yielding $f_2(X_i)$ as the influence function. Therefore
    \[
        \psi^{it}(X_i) = -(A+B)^{-1}f_2(X_i),
    \]
    matching Proposition~\ref{prp:gmm_ifs}(iv). \qedhere
\end{proof}

\begin{proof}[Proof of Corollary \ref{cor:miss}]
    $(M_\nu)_k\nu_i\in\mathcal L(\nu_i)$, so $\Pi_\nu^\perp\psi_k(X_i) = \Pi_\nu^\perp r_k(X_i)$. By the definition of $\Delta_k$ as an $R^2$,
    \[
        \Delta_k = 1 - \frac{\mathbb E[(\Pi_\nu^\perp\psi_k(X_i))^2]}{\mathbb E[\psi_k(X_i)^2]} = 1 - \frac{\mathbb E[(\Pi_\nu^\perp r_k(X_i))^2]}{\mathbb E[\psi_k(X_i)^2]}.
    \]
\end{proof}

\begin{proof}[Proof of Proposition \ref{prp:eff_channel}]
    Since $f$ is nonrandom,
    $\partial g(X_i,\theta)/\partial\theta'=-\partial f(\theta)/\partial\theta'$
    is nonrandom, so $\gamma_i=0$ and the Jacobian channel vanishes; and
    $\widehat W=\widehat V^{-1}$ depends on the data only through $\{m_i\}$, not
    through a preliminary estimate of $\theta$, so the first-step channel is
    absent. The sample covariance $\widehat V$ has influence function
    $\nu_i\nu_i'-V$, so by the delta method for matrix inversion,
    $\mathrm d(M^{-1})=-M^{-1}(\mathrm dM)M^{-1}$ at $M=V$, the weight has
    influence function $-W(\nu_i\nu_i'-V)W=W-r_ir_i'$ with $r_i=W\nu_i$, i.e.\
    $\omega_i=\vect(W-r_ir_i')$. Using $G'Wg=0$,
    \[
        (g'\otimes G')\omega_i=G'(W-r_ir_i')g=-(G'r_i)(r_i'g)
        =-(G'W\nu_i)(\nu_i'Wg),
    \]
    and Proposition~\ref{prp:gmm_ifs}(ii) gives
    $\psi(X_i)=-A^{-1}[\,G'W\nu_i+(g'\otimes G')\omega_i\,]
    =M_\nu\nu_i(1-\nu_i'Wg)$.

    From $WVW=W$:
    $\mathbb E[(M_\nu\nu_i)(M_\nu\nu_i)']=A^{-1}G'WVWGA^{-1}=V_0$,
    $\mathbb E[\nu_i'Wg]=0$, $\mathbb E[(\nu_i'Wg)^2]=g'WVWg=g'Wg$, and
    $\mathbb E[(M_\nu\nu_i)(\nu_i'Wg)]=M_\nu VWg=-A^{-1}G'Wg=0$ by the
    first-order condition. Hence $\mathbb E[(1-\nu_i'Wg)^2]=1+g'Wg$,
    $\mathbb E[\psi(X_i)]=0$, and, by the definition of covariance,
    \[
        \operatorname{Var}[\psi_k(X_i)]
        =\mathbb E\big[(M_\nu\nu_i)_k^2(1-\nu_i'Wg)^2\big]
        =\mathbb E[(M_\nu\nu_i)_k^2]\,\mathbb E[(1-\nu_i'Wg)^2]+\kappa_k
        =(1+g'Wg)\,V_{0,kk}+\kappa_k .
    \]
    In the notation of Corollary~\ref{cor:miss},
    $\psi_k(X_i)=(M_\nu)_k\nu_i+r_k(X_i)$ with
    $r_k(X_i)=-(M_\nu\nu_i)_k(\nu_i'Wg)$, so $\Delta_k\le1$ always; if $g=0$ then
    $\nu_i'Wg\equiv0$, $r_k\equiv0\in\mathcal L(\nu_i)$, and $\Delta_k=1$.

    If $\nu_i$ is Gaussian, the vector $G'W\nu_i$ and the scalar $\nu_i'Wg$ are
    jointly Gaussian and uncorrelated,
    $\mathbb E[(G'W\nu_i)(\nu_i'Wg)]=G'Wg=0$, hence independent. Then
    $(M_\nu\nu_i)(M_\nu\nu_i)'$ is independent of $(1-\nu_i'Wg)^2$, so
    $\kappa_k=0$ and
    $\operatorname{Var}[\psi(X_i)]
    =\mathbb E[(M_\nu\nu_i)(M_\nu\nu_i)']\,\mathbb E[(1-\nu_i'Wg)^2]
    =(1+g'Wg)V_0$. Because all third moments of the mean-zero Gaussian $\nu_i$
    vanish, $\mathbb E[r_k(X_i)\,c'\nu_i]=0$ for every $c$, so
    $\Pi_\nu^\perp r_k=r_k$ and, by Corollary~\ref{cor:miss},
    \[
        \Delta_k=1-\frac{\mathbb E[r_k(X_i)^2]}{\mathbb E[\psi_k(X_i)^2]}
        =1-\frac{(g'Wg)\,V_{0,kk}}{(1+g'Wg)\,V_{0,kk}}=\frac{1}{1+g'Wg}.
    \]
\end{proof}

\begin{proof}[Proof of Proposition \ref{prp:sens_to_1s}]
    Let \(F_{n}(\theta, \phi) = \widehat{G}(\theta)^{\prime} \widehat{W}(\phi) \widehat{g}(\theta)\) denote the sample FOC function for the second-step GMM estimator. By definition, the second-step estimator \(\widehat{\theta}(\widehat{\phi})\) exactly satisfies the sample FOC evaluated at the first-step estimate \(\widehat{\phi}\):
    \[
        F_{n}(\widehat{\theta}(\widehat{\phi}), \widehat{\phi}) = 0.
    \]

    Under Assumption \ref{ass:regularity}, the sample moments and their derivatives are continuously differentiable in a neighborhood of the estimates. Applying the Implicit Function Theorem to the FOC with respect to \(\widehat{\phi}\) yields the finite-sample derivative
    \[
        \frac{\partial \widehat{\theta}(\widehat{\phi})}{\partial \widehat{\phi}^{\prime}} = -[F_{n,\theta}(\widehat{\theta}, \widehat{\phi})]^{-1} F_{n,\phi}(\widehat{\theta}, \widehat{\phi}),
    \]
    where \(F_{n,\theta}(\theta, \phi) = \frac{\partial F_n(\theta, \phi)}{\partial \theta^{\prime}}\), \(F_{n,\phi}(\theta, \phi) = \frac{\partial F_n(\theta, \phi)}{\partial \phi^{\prime}}\), and \(\widehat{\theta}\) is the second-step GMM estimator.

    As established in the proof of Proposition \ref{prp:gmm_ifs}, the ULLN over the compact parameter space combined with the consistency of the preliminary and second-step estimators, \((\widehat{\theta}, \widehat{\phi}) \xrightarrow{p} (\theta_0, \phi_0)\), ensures that the sample derivatives evaluated at the estimates converge in probability to their population counterparts evaluated at the pseudo-true values:
    \[
        F_{n,\theta}(\widehat{\theta}, \widehat{\phi}) \xrightarrow{p} A \quad \text{and} \quad F_{n,\phi}(\widehat{\theta}, \widehat{\phi}) \xrightarrow{p} B,
    \]
    where \(A = G^{\prime} W(\phi_0) G + H\) is the population curvature matrix of the second-step objective, and \(B = (g^{\prime} \otimes G^{\prime}) S(\phi_0)\) captures the sensitivity of the optimal weight matrix to the first-step parameter.

    By Assumption \ref{ass:regularity}(iv), the population curvature matrix \(A\) is nonsingular. Applying the Continuous Mapping Theorem to the matrix inverse, we obtain the probability limit of the finite-sample sensitivity
    \[
        \Lambda_\phi = \text{plim} \frac{\partial \widehat{\theta}(\widehat{\phi})}{\partial \widehat{\phi}^{\prime}} = -A^{-1}B.
    \]
\end{proof}

The following lemma records the geometric convergence of the iterated
estimand, used in the proof of Proposition~\ref{prp:variance_reduction}.

\begin{lemma}[Geometric convergence of the iterated estimand]
    \label{lem:estimand}
    Suppose Assumption~\ref{ass:regularity} holds, and let $\theta^{(s)}$ be
    the probability limit of the $s$-step GMM estimator, defined recursively
    by the population first-order condition
    $G(\theta^{(s)})'W(\theta^{(s-1)})g(\theta^{(s)})=0$. For every
    $\rho\in(\rho(-A^{-1}B),1)$ there exist a neighborhood $\mathcal N_0$ of
    the iterated fixed point $\theta_0$ and a constant $c_0>0$ such that, for
    any first step $\theta^{(0)}\in\mathcal N_0$,
    $\|\theta^{(s)}-\theta_0\|\le c_0\,\rho^{\,s}$ for all $s\ge0$.
\end{lemma}

\begin{proof}
    By Assumption~\ref{ass:regularity}(iii), (iv), and (vi)(b), the implicit
    function theorem applied to $G(\theta)'W(\phi)g(\theta)=0$ defines the
    population updating map $T(\phi)$ continuously differentiably near
    $\theta_0$, with $T'(\theta_0)=-A^{-1}B$
    (Proposition~\ref{prp:sens_to_1s}). Fix $\rho\in(\rho(-A^{-1}B),1)$;
    applying \citet*[Fact~9.8.4]{bernstein2009matrix} to
    $\rho^{-1}(-A^{-1}B)$, whose spectral radius is below one, gives a
    submultiplicative matrix norm $\|\cdot\|_*$ with $\|-A^{-1}B\|_*<\rho$,
    and we equip $\mathbb R^p$ with a compatible vector norm, also written
    $\|\cdot\|_*$. By continuous differentiability there is a neighborhood
    $\mathcal N_0$ of $\theta_0$ on which
    $\|T(\theta)-\theta_0\|_*\le\rho\|\theta-\theta_0\|_*$; since $\rho<1$,
    $\mathcal N_0$ is forward invariant and iterating from
    $\theta^{(0)}\in\mathcal N_0$ gives
    $\|\theta^{(s)}-\theta_0\|_*\le\rho^{\,s}\|\theta^{(0)}-\theta_0\|_*$.
    The claim follows by equivalence of norms on $\mathbb R^p$, with $c_0$
    determined by the norm-equivalence constants. This records, in the form
    used below, the contraction result of \citet*{hansen2021inference}.
\end{proof}

\begin{proof}[Proof of Proposition~\ref{prp:variance_reduction}]
    Fix $\tilde\rho\in(\rho(-A^{-1}B),1)$ and
    $\rho_1\in(\rho(-A^{-1}B),\tilde\rho)$. Let $\psi^{(s)}$ denote the
    idealized recursion with all matrices fixed at $\theta_0$,
    \[
        \psi^{(s)}(X_i)=-A^{-1}f_2(X_i)-A^{-1}B\,\psi^{(s-1)}(X_i),\qquad
        \psi^{(1)}(X_i)=\psi^{1s,W}(X_i)=-A^{-1}f_2(X_i),
    \]
    with $f_2$ as in the proof of Proposition~\ref{prp:gmm_ifs} evaluated at
    $\theta_0$. Since $\rho(-A^{-1}B)<1$ by
    Assumption~\ref{ass:regularity}(v), the Neumann series gives
    $\psi^{it}=-(A+B)^{-1}f_2=\sum_{k\ge0}(-A^{-1}B)^k(-A^{-1}f_2)$, so
    $\psi^{it}$ is the fixed point of the recursion, the difference
    $d_s=\psi^{(s)}-\psi^{it}$ satisfies $d_s=(-A^{-1}B)\,d_{s-1}$, and
    \[
        \|\psi^{(s)}(X_i)-\psi^{it}(X_i)\|_{L^2(P)}
        \le\|(-A^{-1}B)^{s-1}\|_{\mathrm{op}}\,
        \|\psi^{1s,W}(X_i)-\psi^{it}(X_i)\|_{L^2(P)}
        \le c_1\tilde\rho^{\,s-1}\,
        \|\psi^{1s,W}(X_i)-\psi^{it}(X_i)\|_{L^2(P)},
    \]
    where $\|(-A^{-1}B)^{k}\|_{\mathrm{op}}\le c_1\tilde\rho^{\,k}$ follows
    from applying \citet*[Fact~9.8.4]{bernstein2009matrix} to
    $\tilde\rho^{-1}(-A^{-1}B)$, whose spectral radius is below one, together
    with submultiplicativity of the resulting norm $\|\cdot\|_*$ and
    equivalence of norms.

    For the evaluation error $e_s=\psi^{ss}-\psi^{(s)}$, subtracting the two
    recursions gives
    \[
        e_s=(-A_s^{-1}B_s)\,e_{s-1}+\mathcal R_s,\qquad
        \mathcal R_s=\big(A^{-1}f_2-A_s^{-1}f_{2,s}\big)
        +\big(A^{-1}B-A_s^{-1}B_s\big)\psi^{(s-1)} .
    \]
    Since $\theta^{(0)}\in\mathcal N_0$, Lemma~\ref{lem:estimand} gives
    $\|\theta^{(s)}-\theta_0\|=O(\rho_1^{s})$, and by
    Assumption~\ref{ass:regularity}(iii) and the Lipschitz weight conditions
    the maps $(\theta,\phi)\mapsto A^{-1}B$ and
    $(\theta,\phi)\mapsto A^{-1}f_2$ are locally Lipschitz near
    $(\theta_0,\theta_0)$, the latter in $L^2(P)$. Hence
    $\|\mathcal R_s\|_{L^2(P)}=O(\rho_1^{\,s-1})$, $e_1$ is bounded in
    $L^2(P)$, and $\|(-A_k^{-1}B_k)-(-A^{-1}B)\|_*=O(\rho_1^{\,k-1})$, so the
    products of the step-specific matrices satisfy
    \[
        \Big\|\prod_{k=j+1}^{s}(-A_k^{-1}B_k)\Big\|_*
        \le\prod_{k=j+1}^{s}\big(\tilde\rho+O(\rho_1^{\,k-1})\big)
        \le\tilde\rho^{\,s-j}\prod_{k\ge1}
        \big(1+O(\rho_1^{\,k-1})/\tilde\rho\big)
        \le c_2\,\tilde\rho^{\,s-j},
    \]
    the infinite product converging because $\sum_k\rho_1^{k}<\infty$.
    Unrolling the $e_s$ recursion,
    \[
        \|e_s\|_{L^2(P)}
        \le c_2\tilde\rho^{\,s-1}\|e_1\|_{L^2(P)}
        +c_2\sum_{j=2}^{s}\tilde\rho^{\,s-j}\,\|\mathcal R_j\|_{L^2(P)}
        \le c_3\Big(\tilde\rho^{\,s-1}
        +\sum_{j=2}^{s}\tilde\rho^{\,s-j}\rho_1^{\,j-1}\Big)
        \le C_{\tilde\rho}\,\tilde\rho^{\,s-1},
    \]
    the last step because
    $\sum_{j=2}^{s}\tilde\rho^{\,s-j}\rho_1^{\,j-1}
    =\tilde\rho^{\,s-1}\sum_{j=2}^{s}(\rho_1/\tilde\rho)^{\,j-1}
    \le\tilde\rho^{\,s-1}\,\rho_1/(\tilde\rho-\rho_1)$, the geometric series
    in $\rho_1/\tilde\rho$ being convergent. The triangle inequality combines
    the two bounds and yields the stated rate.

    For the consequences, Cauchy--Schwarz gives
    $|\operatorname{Var}(\psi^{ss}_k)-\operatorname{Var}(\psi^{it}_k)|
    \le\|\psi^{ss}_k-\psi^{it}_k\|_{L^2(P)}
    \,(\|\psi^{ss}_k\|_{L^2(P)}+\|\psi^{it}_k\|_{L^2(P)})
    =O(\tilde\rho^{\,s-1})$. For the informativeness, write
    $\nu_i(\theta)=g(X_i,\theta)-\mathbb E[g(X_i,\theta)]$; the projection
    direction drifts with the step, but
    $\|\nu_i(\theta^{(s)})-\nu_i(\theta_0)\|_{L^2(P)}=O(\rho_1^{\,s})$ by the
    mean-square Lipschitz continuity of $\theta\mapsto g(X_i,\theta)$ under
    Assumption~\ref{ass:regularity}(iii), so
    $\mathbb E[\psi^{ss}\nu_i(\theta^{(s)})']$ and
    $\mathbb E[\nu_i(\theta^{(s)})\nu_i(\theta^{(s)})']$ converge to their
    fixed-point counterparts at the rate $\tilde\rho^{\,s-1}$. The limiting
    second-moment matrix is nonsingular, as required throughout for the
    sensitivity $\Lambda$, and matrix inversion is locally Lipschitz at a
    nonsingular matrix, so the projection coefficients and the projected
    variance converge at the same rate; when
    $\operatorname{Var}[\psi^{it}_k(X_i)]>0$ the ratio converges,
    $\Delta_k^{ss}\to\Delta_k^{it}$, at the rate $\tilde\rho^{\,s-1}$.
\end{proof}

\newpage
\section{Appendix: CUGMM}
\label{sec:cugmm}

The Continuously Updating GMM (CUGMM) estimator treats the weight matrix as a function of the parameter and updates it jointly with \(\theta\). The parameter dependence of the weight matrix complicates the asymptotic analysis under misspecification. This appendix provides a formal treatment of the CUGMM estimator. We establish its consistency, derive its misspecification-robust influence function, and provide the explicit derivations for the \citet*{schennach2007point} example.

The CUGMM estimator is defined as the minimizer of the sample objective function
\[
    \widehat{Q}_{cue}(\theta) = \widehat{g}(\theta)'\widehat{\Omega}(\theta)^{-1}\widehat{g}(\theta),
\]
where \(\widehat{\Omega}(\theta) = \frac{1}{n}\sum_{i=1}^n g(X_i, \theta)g(X_i, \theta)'\) is the uncentered sample second-moment matrix.

The CUGMM objective does not depend on whether the weight matrix uses the uncentered second-moment matrix \(\widehat{\Omega}(\theta)\) or the centered covariance \(\widehat{V}(\theta) = \widehat{\Omega}(\theta) - \widehat{g}(\theta)\widehat{g}(\theta)'\). By the Sherman-Morrison formula, \(\widehat{Q}_{cue}(\theta) = \widehat{Q}_c(\theta)/(1 + \widehat{Q}_c(\theta))\), where \(\widehat{Q}_c(\theta) = \widehat{g}(\theta)'\widehat{V}(\theta)^{-1}\widehat{g}(\theta)\), and \(x/(1+x)\) is strictly increasing for \(x \ge 0\). The two objectives are monotone transformations of each other and share the same minimizer, in sample and in population, so \(\widehat{\theta}_{cue}\) and its pseudo-true value \(\theta_0\) do not depend on the centering. We use the uncentered matrix \(\Omega(\theta)\) without loss of generality.

To characterize the local behavior of CUGMM, we define the following higher-order population curvature matrices evaluated at \(\theta_0\). Let \(W = \Omega(\theta_0)^{-1}\) and \(v = Wg\). Furthermore, let \(S = \frac{\partial}{\partial \theta'} \text{vec}(\Omega(\theta))\big|_{\theta_0}\). We define
\begin{align*}
    T & = \frac{\partial}{\partial \theta'} \text{vec}(S(\theta)')\Big|_{\theta_0},                        \\
    K & = \left((v \otimes v)' \otimes I_p\right)T + 2S'(v \otimes I_q)\left(WG - (g'W \otimes W)S\right).
\end{align*}
The misspecification-robust population curvature matrix for CUGMM is defined as \(A_{cu} = 2(A + B) - K\), where \(A\) and \(B\) are the population curvature matrices of iterated GMM evaluated at \(W\), as defined in Section \ref{sec:ifgmm}.

We impose the following regularity conditions:
\begin{assumption}
    \label{ass:cugmm}
    \begin{enumerate}
        \item[(i)] There exists a unique \(\theta_0\) in the interior of \(\Theta\) such that for any \(\epsilon > 0\), \(\inf_{\theta \in \Theta: \|\theta - \theta_0\| \ge \epsilon} Q_{cue}(\theta) > Q_{cue}(\theta_0)\), where \(Q_{cue}(\theta) = \mathbb{E}[g(X_i,\theta)]'\Omega(\theta)^{-1}\mathbb{E}[g(X_i,\theta)]\).
        \item[(ii)] The population curvature matrix \(A_{cu}\) is nonsingular.
        \item[(iii)] For the population centered covariance matrix \(V(\theta) = \Omega(\theta) - \mathbb{E}[g(X_i,\theta)]\mathbb{E}[g(X_i,\theta)]'\), there exists a constant \(c > 0\) such that \(\inf_{\theta \in \Theta} \lambda_{\min}(V(\theta)) \ge c\).
    \end{enumerate}
\end{assumption}

Assumption \ref{ass:cugmm}(iii) precludes the variance collapse highlighted by \citet*{kleibergen2025double}. By ensuring uniform non-singularity, we guarantee that the estimator does not spuriously minimize the objective by pushing the parameter to the boundary and inflating the variance determinant. Assumption \ref{ass:cugmm}(i) ensures that identification strength dominates the degree of misspecification.

\begin{proposition}[Consistency of CUGMM under Misspecification]
    \label{prp:cugmm_consistency}
    Suppose Assumption \ref{ass:regularity} and Assumption \ref{ass:cugmm} hold. Then, the CUGMM estimator converges in probability to the unique pseudo-true value \(\widehat{\theta}_{cue} \xrightarrow{p} \theta_0\).
\end{proposition}

\begin{proof}[Proof of Proposition \ref{prp:cugmm_consistency}]
    First, we show that Assumption \ref{ass:cugmm}(i) implies \(\theta_0\) is the unique minimizer of the population objective \(Q_{cue}(\theta)\). Suppose, for a contradiction, there exists another minimizer \(\theta_1 \neq \theta_0\). Then \(Q_{cue}(\theta_1) \le Q_{cue}(\theta_0)\). Let \(\epsilon = \|\theta_1 - \theta_0\| > 0\). By Assumption~\ref{ass:cugmm}(i), \(\inf_{\|\theta - \theta_0\| \ge \epsilon} Q_{cue}(\theta) > Q_{cue}(\theta_0)\). Since \(\|\theta_1 - \theta_0\| \ge \epsilon\), it follows that \(Q_{cue}(\theta_1) \ge \inf_{\|\theta - \theta_0\| \ge \epsilon} Q_{cue}(\theta) > Q_{cue}(\theta_0)\), which contradicts \(Q_{cue}(\theta_1) \le Q_{cue}(\theta_0)\). Thus, \(\theta_0\) is unique.

    By the uniform integrability conditions in Assumption \ref{ass:regularity}, the ULLN applies to the sample moments and the uncentered sample covariance matrix over the compact set \(\Theta\):
    \[
        \sup_{\theta \in \Theta} \|\widehat{g}(\theta) - \mu(\theta)\| \xrightarrow{p} 0, \quad \text{and} \quad \sup_{\theta \in \Theta} \|\widehat{\Omega}(\theta) - \Omega(\theta)\| \xrightarrow{p} 0,
    \]
    where \(\mu(\theta) = \mathbb{E}[g(X_i, \theta)]\).

    Because \(\Omega(\theta) = V(\theta) + \mu(\theta)\mu(\theta)'\), we have \(\Omega(\theta) \ge V(\theta)\) in the positive semi-definite sense. Therefore, Assumption \ref{ass:cugmm}(iii) implies \(\inf_{\theta \in \Theta} \lambda_{\min}(\Omega(\theta)) \ge c > 0\) holds as well.

    Combined with the uniform convergence \(\sup_{\theta}\|\widehat{\Omega}(\theta) - \Omega(\theta)\| \xrightarrow{p} 0\), this lower bound implies that, with probability approaching one, \(\lambda_{\min}(\widehat{\Omega}(\theta)) \ge c/2\) for all \(\theta \in \Theta\). On the set of symmetric matrices with smallest eigenvalue at least \(c/2\), the inverse satisfies the Lipschitz bound \(\|\widehat{\Omega}(\theta)^{-1} - \Omega(\theta)^{-1}\| = \|\widehat{\Omega}(\theta)^{-1}(\Omega(\theta) - \widehat{\Omega}(\theta))\Omega(\theta)^{-1}\| \le (2/c^2)\|\widehat{\Omega}(\theta) - \Omega(\theta)\|\). Taking the supremum over \(\theta\) gives
    \[
        \sup_{\theta \in \Theta} \|\widehat{\Omega}(\theta)^{-1} - \Omega(\theta)^{-1}\| \le \frac{2}{c^2}\,\sup_{\theta \in \Theta} \|\widehat{\Omega}(\theta) - \Omega(\theta)\| \xrightarrow{p} 0.
    \]
    We bound the difference between the sample and population objective functions using the triangle inequality and sub-multiplicativity of norms:
    \begin{align*}
        |\widehat{Q}_{cue}(\theta) - Q_{cue}(\theta)| \le & \ \|\widehat{g}(\theta) - \mu(\theta)\|^2 \|\widehat{\Omega}(\theta)^{-1}\| \nonumber                \\
                                                          & + 2\|\mu(\theta)\| \|\widehat{\Omega}(\theta)^{-1}\| \|\widehat{g}(\theta) - \mu(\theta)\| \nonumber \\
                                                          & + \|\mu(\theta)\|^2 \|\widehat{\Omega}(\theta)^{-1} - \Omega(\theta)^{-1}\|.
    \end{align*}
    Since \(\mu(\theta)\) and \(\Omega(\theta)^{-1}\) are uniformly bounded over \(\Theta\), each term converges to zero uniformly in probability. Thus, \(\sup_{\theta \in \Theta} |\widehat{Q}_{cue}(\theta) - Q_{cue}(\theta)| \xrightarrow{p} 0\).

    By definition, the estimator satisfies \(\widehat{Q}_{cue}(\widehat{\theta}_{cue}) \le \widehat{Q}_{cue}(\theta_0)\). Using the uniform convergence established above, it follows that \(Q_{cue}(\widehat{\theta}_{cue}) \le Q_{cue}(\theta_0) + o_p(1)\). Finally, for any \(\epsilon > 0\), there exists an \(\eta > 0\) such that \(\inf_{\|\theta - \theta_0\| \ge \epsilon} Q_{cue}(\theta) > Q_{cue}(\theta_0) + \eta\). The event \(\|\widehat{\theta}_{cue} - \theta_0\| \ge \epsilon\) implies \(\{Q_{cue}(\widehat{\theta}_{cue}) > Q_{cue}(\theta_0) + \eta\}\). However, since \(Q_{cue}(\widehat{\theta}_{cue}) \le Q_{cue}(\theta_0) + o_p(1)\), the probability of the latter converges to zero. Therefore, \(\lim_{n\rightarrow\infty} \mathbb{P}(\|\widehat{\theta}_{cue} - \theta_0\| \ge \epsilon) = 0\), establishing that \(\widehat{\theta}_{cue} \xrightarrow{p} \theta_0\).
\end{proof}

\begin{proposition}[Influence Function for CUGMM]
    \label{prp:cugmm_if}
    Suppose the conditions of Proposition \ref{prp:cugmm_consistency} hold. The influence function for the CUGMM estimator under misspecification is given by
    \[
        \psi^{cu}(X_i) = -A_{cu}^{-1} f_{cu}(X_i),
    \]
    where \(f_{cu}(X_i) = f_1(X_i) - f_2(X_i)\). Letting \(v = Wg\), \(\psi_\Omega(X_i) = g(X_i,\theta_0)g(X_i,\theta_0)' - \Omega(\theta_0)\), and \(\widetilde{\psi}_v(X_i) = W\nu(X_i) - W\psi_\Omega(X_i)v\), the components are defined as
    \begin{align*}
        f_1(X_i) & = 2(G(X_i,\theta_0) - G)'v + 2G'\widetilde{\psi}_v(X_i),                                                             \\
        f_2(X_i) & = \widetilde{\psi}_{Sv}(X_i) + S'\left(v \otimes \widetilde{\psi}_v(X_i) + \widetilde{\psi}_v(X_i) \otimes v\right),
    \end{align*}
    with \(\widetilde{\psi}_{Sv}(X_i) = 2(g(X_i,\theta_0)'v)G(X_i,\theta_0)'v - S'(v \otimes v)\).
\end{proposition}

\begin{proof}[Proof of Proposition \ref{prp:cugmm_if}]
    The CUGMM sample FOC is obtained by taking the total derivative of the objective function. By the chain rule and the identity \(\frac{\partial \text{vec}(\widehat{W}(\theta))}{\partial \theta'} = -(\widehat{W}(\theta) \otimes \widehat{W}(\theta))\widehat{S}(\theta)\), the FOC is
    \begin{align*}
        F_{n}(\theta) & = 2\widehat{G}(\theta)'\widehat{W}(\theta)\widehat{g}(\theta) - \widehat{S}(\theta)'(\widehat{W}(\theta) \otimes \widehat{W}(\theta))'(\widehat{g}(\theta) \otimes \widehat{g}(\theta)) \\
                      & = 2\widehat{G}(\theta)'\widehat{W}(\theta)\widehat{g}(\theta) - \widehat{S}(\theta)'(\widehat{W}(\theta)\widehat{g}(\theta) \otimes \widehat{W}(\theta)\widehat{g}(\theta)).
    \end{align*}
    Letting \(\widehat{v}(\theta) = \widehat{W}(\theta)\widehat{g}(\theta)\), the estimator satisfies \(F_{n}(\widehat{\theta}_{cue}) = 0\). At \(\theta_0\), the population FOC is \(F(\theta_0) = 2G'v - S'(v \otimes v) = 0\).

    Applying the mean value theorem around \(\theta_0\) gives \(0 = F_n(\theta_0) + F_{n, \theta}(\widetilde{\theta})(\widehat{\theta}_{cue} - \theta_0)\). Differentiating the FOC once more in \(\theta\), the linear term \(2\widehat{G}(\theta)'\widehat{W}(\theta)\widehat{g}(\theta)\) contributes \(2(A+B)\) in the limit, exactly the iterated-GMM curvature, while the quadratic term \(\widehat{S}(\theta)'(\widehat{v}(\theta) \otimes \widehat{v}(\theta))\) contributes the additional curvature \(K\) through the derivatives of \(\widehat{S}(\theta)\) and \(\widehat{W}(\theta) \otimes \widehat{W}(\theta)\). By the uniform convergence results of Proposition~\ref{prp:cugmm_consistency} and the consistency of \(\widetilde{\theta}\), the continuous mapping theorem gives \(F_{n, \theta}(\widetilde{\theta}) \xrightarrow{p} A_{cu} = 2(A+B) - K\).

    Next, we expand \(\sqrt{n}F_{n}(\theta_0)\) around population limits \(G\), \(v\), and \(S\). Let \(\widehat{G} = G + (\widehat{G}-G)\), \(\widehat{v} = v + (\widehat{v}-v)\), and \(\widehat{S} = S + (\widehat{S}-S)\). Substituting these into the sample FOC and collecting first-order terms yields
    \begin{align*}
        F_{n}(\theta_0) & = 2\left[ (\widehat{G}-G)'v + G'(\widehat{v}-v) \right]                                                                                    \\
                        & \quad - \left[ (\widehat{S}-S)'(v \otimes v) + S'\left(v \otimes (\widehat{v}-v) + (\widehat{v}-v) \otimes v\right) \right] + O_p(n^{-1}).
    \end{align*}
    We seek the influence function \(f_{cu}(X_i)\) such that \(\sqrt{n}F_n(\theta_0) = \frac{1}{\sqrt{n}}\sum_i f_{cu}(X_i) + o_p(1)\). We define \(f_{cu}(X_i) = f_1(X_i) - f_2(X_i)\) corresponding to the two bracketed terms, where \(f_1\) and \(f_2\) are the CUGMM-specific quantities given in the proposition statement, not the like-named terms in the proof of Proposition~\ref{prp:gmm_ifs}.

    To find \(f_{cu}(X_i)\), we determine the influence function for the composite vector \(\widehat{v}\). Since \(\sqrt{n}(\widehat{v}-v) = W\sqrt{n}(\widehat{g}-g) + \sqrt{n}(\widehat{W}-W)g + o_p(1)\), and by the delta method \(\sqrt{n}(\widehat{W}-W)\) has influence function \(-W\psi_\Omega(X_i)W\), the influence function for \(\sqrt{n}(\widehat{v}-v)\) is \(\widetilde{\psi}_v(X_i) = W\nu(X_i) - W\psi_\Omega(X_i)v\). Substituting \(\widetilde{\psi}_v(X_i)\) into the first bracket yields \(f_1(X_i) = 2(G(X_i,\theta_0) - G)'v + 2G'\widetilde{\psi}_v(X_i)\).

    For the second bracket, the term \(\sqrt{n}(\widehat{S}-S)'(v \otimes v)\) requires the influence function of \(\widehat{S}\), evaluated in projection with \(v \otimes v\). Let \(M_i(\theta_0) = \frac{\partial \text{vec}(g(X_i, \theta)g(X_i, \theta)')}{\partial \theta'}\big|_{\theta_0}\), such that \(\widehat{S} = \frac{1}{n}\sum_i M_i(\theta_0)\). Using the product rule, \(M_i = [(g_i \otimes I_q) + (I_q \otimes g_i)]G_i\), where \(g_i = g(X_i, \theta_0)\) and \(G_i = G(X_i, \theta_0)\). Post-multiplying by \((v \otimes v)\) gives
    \[
        M_i'(v \otimes v) = 2(g_i'v)G_i'v.
    \]
    Demeaning this gives the influence function \(\widetilde{\psi}_{Sv}(X_i) = 2(g_i'v)G_i'v - S'(v \otimes v)\). Combining these components yields \(f_2(X_i) = \widetilde{\psi}_{Sv}(X_i) + S'\left(v \otimes \widetilde{\psi}_v(X_i) + \widetilde{\psi}_v(X_i) \otimes v\right)\).

    Collecting all terms, the influence function of the CUGMM estimator is \(\psi^{cu}(X_i) = -A_{cu}^{-1}(f_1(X_i) - f_2(X_i))\).
\end{proof}

\newpage
\section{Appendix: Derivation of Examples}

\subsection{Derivation of Example \ref{exm:schennach}}
\label{sec:schennach_deriv}

This appendix details the derivations for Example \ref{exm:schennach}. We consider the moment function \(g(X_i, \theta) = (X_i - \theta, (X_i - \theta)^2 - 1)'\), assuming the data \(X_i\) are i.i.d. draws from a normal distribution \(N(\mu, \sigma^2)\). The parameter of interest is the mean \(\theta\), while the variance is constrained to one. Let \(U_i = X_i - \mu\) denote the demeaned data, \(\delta = \mu - \theta\) the parameter bias, and \(a = \sigma^2 - 1\) the degree of misspecification. The model is correctly specified if \(\sigma^2=1\) (i.e., \(a=0\)).

We first establish the population quantities required by Proposition \ref{prp:gmm_ifs}. The expected moment vector is \(g = [0, a]'\), and the Jacobian is \(G = [-1, 0]'\). The population Hessian of the moments is \(R = \mathbb{E}[\partial \vect(G(X_i, \theta)')/\partial\theta']|_{\mu} = [0, 2]'\). The covariance matrix of the moments is \(\Omega(\mu) = \text{diag}(\sigma^2, \tau)\), where \(\tau = \mathbb{E}[(U_i^2 - 1)^2] = 3\sigma^4 - 2\sigma^2 + 1\). The efficient weight matrix is \(W = \text{diag}(1/\sigma^2, 1/\tau)\). We also define \(S = \partial \vect(\Omega(\theta))/\partial \theta' |_{\mu}\). Noting that \(\partial/\partial\theta = -\partial/\partial\delta\), we calculate the derivatives of the covariance terms at \(\mu\), yielding \(S = [0, s_{12}, s_{12}, 0]'\) with \(s_{12} = 1-3\sigma^2\).

\textbf{One-Step GMM (W=I)}

The one-step GMM estimator minimizes the objective function \(Q_{I}(\theta) = g(\theta)'g(\theta) = \delta^2 + (a+\delta^2)^2\). The first-order condition with respect to \(\delta\) is:
\begin{align*}
    \frac{\partial Q_I}{\partial \delta} = 2\delta + 4\delta(a+\delta^2) = 2\delta(1 + 2a + 2\delta^2) = 0.
\end{align*}
Substituting \(a = \sigma^2 - 1\), the condition becomes \(2\delta(2\sigma^2 - 1 + 2\delta^2) = 0\). Therefore,
\[
    \theta_0 = \begin{cases}
        \mu,                                   & \text{if } \sigma^2 \geq \frac{1}{2}, \\
        \mu \pm \sqrt{\frac{1}{2} - \sigma^2}, & \text{if } \sigma^2 < \frac{1}{2}.
    \end{cases}
\]

We analyze the sensitivity and informativeness at \(\theta_0=\mu\) under the condition \(\sigma^2 \geq 1/2\). The curvature components at \(\theta_0=\mu\) are \(H = 2a\) and \(A = G'IG + H = 2\sigma^2 - 1\). Applying Proposition \ref{prp:gmm_ifs}(i), the influence function is \(\psi^{1s}(X_i) = -A^{-1}(G'Ig(X_i, \mu) + G(X_i, \mu)'Ig) = U_i\). It follows that \(\Lambda^{1s} = [1, 0]\) and \(\Delta^{1s} = 1\).

\textbf{Two-Step GMM}

We assume the first-step estimator is the one-step GMM estimator, which requires \(\sigma^2 \geq 1/2\) to converge to \(\mu\). The second-step estimator minimizes \(Q_{2S}(\theta) = g(\theta)'W(\mu)g(\theta) = \delta^2/\sigma^2 + (a+\delta^2)^2/\tau\). The FOC is
\begin{align*}
    \frac{\partial Q_{2S}}{\partial \delta} = 2\delta\left(\frac{1}{\sigma^2} + \frac{2(a+\delta^2)}{\tau}\right) = 0.
\end{align*}
Evaluating the term in the parenthesis at \(\delta=0\) yields \(A = 1/\sigma^2 + 2a/\tau\). Under normality, \(A = (5\sigma^4 - 4\sigma^2 + 1)/(\sigma^2\tau)\), which is positive for all real \(\sigma^2\). Since \(A > 0\) and \(\delta^2/\tau \geq 0\), the term in the parenthesis is positive, ensuring \(\delta=0\) is the unique solution. Thus, the pseudo-true value is \(\theta_0 = \mu\).

We derive the influence function using the common curvature matrices for efficient GMM: \(H = 2a/\tau\), \(B = as_{12}/(\sigma^2\tau)\), and \(A_{it} = 2\sigma^2/\tau\). Using Proposition \ref{prp:gmm_ifs}(iii) with first-step influence function \(\psi^{\phi}(X_i) = U_i\), we obtain
\begin{align*}
    \psi^{2s}(X_i) = U_i - A^{-1}\frac{a}{\sigma^2\tau}Z_i,
\end{align*}
where \(Z_i = U_i^3 - 3\sigma^2 U_i\). Since \(\mathbb{E}[U_i Z_i] = 0\), the sensitivity remains \(\Lambda^{2s} = [1, 0]\). However, the variance increases, leading to \(\Delta^{2s} < 1\) whenever \(a \neq 0\).

\textbf{Iterated GMM}

The iterated GMM estimator is defined by the fixed point condition \(G(\theta)'W(\theta)g(\theta) = 0\). At \(\theta = \mu\), we have \([-1, 0] \cdot \text{diag}(1/\sigma^2, 1/\tau) \cdot [0, a]' = 0\), so \(\mu\) satisfies the condition. For convergence and uniqueness locally, we verify the contraction mapping condition \(|A^{-1}B| < 1\). Under normality, this inequality takes the form:
\begin{align*}
    \left| \frac{a(1-3\sigma^2)}{\sigma^2\tau} \right| < \frac{1}{\sigma^2} + \frac{2a}{\tau} \quad \iff \quad |-3\sigma^4 + 4\sigma^2 - 1| < 5\sigma^4 - 4\sigma^2 + 1.
\end{align*}
This holds for all \(\sigma^2 > 0\). Thus, \(\theta_0 = \mu\) is the unique fixed point.

Using Proposition \ref{prp:gmm_ifs}(iv), the influence function is \(\psi^{it}(X_i) = -A_{it}^{-1}f_{2}(X_i)\), with $f_2$ as in the proof of Proposition~\ref{prp:gmm_ifs}, which simplifies to:
\begin{align*}
    \psi^{it}(X_i) = \frac{U_i}{2\sigma^4} \left[ 5\sigma^4 - 3\sigma^2 - (\sigma^2 - 1)U_i^2 \right].
\end{align*}
We verify \(\mathbb{E}[\psi^{it}U_i] = \sigma^2\), yielding \(\Lambda^{it} = [1, 0]\). The informativeness is \(\Delta^{it} = 2\sigma^4 / (5\sigma^4 - 6\sigma^2 + 3)\), which is less than 1 under misspecification.

\textbf{Continuously Updating GMM}

The CUGMM estimator minimizes \(Q_{cue}(\theta) = g(\theta)'\Omega(\theta)^{-1}g(\theta)\). Note that \(\Omega(\theta) = V(\theta) + g(\theta)g(\theta)'\), where \(V(\theta) = \text{Var}(g(X_i, \theta))\). Substituting into the objective function gives \(Q_{cue}(\theta) = g(\theta)'(V(\theta) + g(\theta)g(\theta)')^{-1}g(\theta)\).

Using the Sherman-Morrison formula,
\begin{align*}
    Q_{cue} & = g'\left(V^{-1} - \frac{V^{-1}gg'V^{-1}}{1 + g'V^{-1}g}\right)g \\
            & = g'V^{-1}g - \frac{g'V^{-1}gg'V^{-1}g}{1 + g'V^{-1}g}           \\
            & = Q_{Var} - \frac{Q_{Var}^2}{1 + Q_{Var}}                        \\
            & = \frac{Q_{Var}}{1 + Q_{Var}}.
\end{align*}

Minimizing \(Q_{cue}(\theta)\) is equivalent to minimizing \(Q_{Var}(\theta) = g(\theta)'V(\theta)^{-1}g(\theta)\), as \(Q_{cue} = Q_{Var} / (1 + Q_{Var})\) is a monotonically increasing function of \(Q_{Var}\).

We calculate \(V(\theta)\) explicitly as a function of \(\delta\). Using the properties of the normal distribution, we obtain
\begin{align*}
    V(\theta) = \begin{pmatrix} \sigma^2 & 2\delta\sigma^2 \\ 2\delta\sigma^2 & 2\sigma^4 + 4\delta^2\sigma^2 \end{pmatrix}.
\end{align*}
The determinant is \(|V(\theta)| = 2\sigma^6\), which is constant and independent of \(\theta\). Substituting \(g(\theta) = [\delta, a+\delta^2]'\) and computing the quadratic form yields:
\begin{align*}
    Q_{Var}(\delta) = \frac{1}{2\sigma^6} \left[ \sigma^2\delta^4 + 2\sigma^2\delta^2(1 - a + \sigma^2) + \sigma^2 a^2 \right].
\end{align*}
Using \(a = \sigma^2 - 1\), the expression simplifies to \(Q_{Var}(\delta) = \frac{1}{2\sigma^4}(\delta^4 + 2\delta^2 + a^2)\). This function is convex in \(\delta^2\) and achieves a unique global minimum at \(\delta=0\) for any \(\sigma^2 > 0\). Thus, the pseudo-true value is uniquely \(\theta_0 = \mu\), provided \(\sigma^2 > 0\).

The curvature matrix of Proposition~\ref{prp:cugmm_if} is \(A_{cu} = 2A_{it} - K\), with \(K\) as defined before Assumption~\ref{ass:cugmm}. Calculating \(K\) under normality yields \(A_{cu} = 8\sigma^4/\tau^2\), which is positive, confirming a local minimum. The influence function derived via Proposition \ref{prp:cugmm_if} is
\begin{align*}
    \psi^{cue}(X_i) = U_i - \frac{a}{2\sigma^2} Z_i.
\end{align*}
Therefore, \(\Lambda^{cue} = [1, 0]\). The informativeness is \(\Delta^{cue} = (1 + \frac{3}{2}(\sigma^2 - 1)^2)^{-1}\), which is less than 1 when \(a \neq 0\).

\subsection{Derivation of Example \ref{exm:schennach2}}
\label{app:schennach2}

We consider the estimation of the mean from i.i.d. data \(X_i \sim N(0, \sigma^2)\). Two researchers use one-step GMM with the identity weight matrix and misspecified moments due to the incorrect assumption that \(X_i \sim N(\theta, 1)\). The first researcher uses the moment function \(g_1(X_i, \theta) = (X_i-\theta, (X_i-\theta)^2-1)'\) and the second researcher uses the moment function \(g_2(X_i, \theta) = (X_i-\theta, (X_i-\theta)^4-3)'\). In both cases, the pseudo-true value is \(\theta_0 = 0\).

For the first researcher, at \(\theta_0=0\), the population moment is:
\[
    g_1 = \mathbb{E}[g_1(X_i, 0)] = \begin{pmatrix} 0 \\ \sigma^2 - 1 \end{pmatrix}.
\]
Let \(a = \sigma^2 - 1\).
The Jacobian is \(G_1(X_i, \theta) = (-1, -2(X_i-\theta))'\). The expected Jacobian at \(\theta_0=0\) is:
\[
    G_1 = \mathbb{E}[G_1(X_i, 0)] = \begin{pmatrix} -1 \\ 0 \end{pmatrix}.
\]
The curvature matrix is \(A_1 = G_1'G_1 + (g_1' \otimes I_1)R_1\). The Hessian of the second moment is \(\mathbb{E}[\nabla_\theta (-2(X_i-\theta))] = 2\). Thus \(H_1 = a(2) = 2a\).
\[
    A_1 = (-1)^2 + 2a = 1 + 2a.
\]

From Proposition \ref{prp:gmm_ifs}(i), the influence function is
\[
    \psi_1(X_i) = -A_1^{-1} f_{1s,1}(X_i),
\]
where \(f_{1s,1}(X_i) = G_1'g_1(X_i, 0) + G_1(X_i, 0)'g_1\), and
\[
    G_1'g_1(X_i, 0) = \begin{pmatrix} -1 & 0 \end{pmatrix} \begin{pmatrix} X_i \\ X_i^2 - 1 \end{pmatrix} = (-1)(X_i) + (0)(X_i^2 - 1) = -X_i,
\]
\[
    G_1(X_i, 0)'g_1 = \begin{pmatrix} -1 & -2X_i \end{pmatrix} \begin{pmatrix} 0 \\ a \end{pmatrix} = (-1)(0) + (-2X_i)(a) = -2aX_i.
\]
Summing these terms yields \(f_{1s,1}(X_i) = -X_i - 2aX_i = -(1+2a)X_i\).
Substituting into the influence function formula
\[
    \psi_1(X_i) = -\frac{1}{1+2a} \left( -(1+2a)X_i \right) = X_i.
\]
The asymptotic variance is \(V_1 = \mathbb{E}[\psi_1(X_i)^2] = \mathbb{E}[X_i^2] = \sigma^2\).
Since \(\psi_1(X_i)\) is perfectly correlated with the first element of \(\nu(X_i)\), \(\Delta_1 = 1\).

For the second researcher, at \(\theta_0=0\),
\[
    g_2 = \mathbb{E}[g_2(X_i, 0)] = \begin{pmatrix} 0 \\ 3\sigma^4 - 3 \end{pmatrix}.
\]
Let \(b = 3(\sigma^4 - 1)\).
The Jacobian is \(G_2(X_i, \theta) = (-1, -4(X_i-\theta)^3)'\). The expected Jacobian is:
\[
    G_2 = \mathbb{E}[G_2(X_i, 0)] = \begin{pmatrix} -1 \\ 0 \end{pmatrix}.
\]
The curvature matrix \(A_2\) involves the expected Hessian of the second moment, \(\mathbb{E}[12X_i^2] = 12\sigma^2\). Thus, \(A_2 = 1 + b(12\sigma^2) = 1 + 12b\sigma^2\).
\[
    G_2'g_2(X_i, 0) = \begin{pmatrix} -1 & 0 \end{pmatrix} \begin{pmatrix} X_i \\ X_i^4 - 3 \end{pmatrix} = -X_i.
\]
\[
    G_2(X_i, 0)'g_2 = \begin{pmatrix} -1 & -4X_i^3 \end{pmatrix} \begin{pmatrix} 0 \\ b \end{pmatrix} = -4bX_i^3.
\]

The influence function is
\[
    \psi_2(X_i) = -A_2^{-1} (-X_i - 4bX_i^3) = \frac{1}{A_2} (X_i + 4bX_i^3).
\]
To calculate the variance \(V_2 = \mathbb{E}[\psi_2^2]\), we use the higher moments of the normal distribution (\(\mathbb{E}[X^4]=3\sigma^4, \mathbb{E}[X^6]=15\sigma^6\)),
\[
    V_2 = \frac{1}{A_2^2} \left( \sigma^2 + 24b\sigma^4 + 240b^2\sigma^6 \right).
\]

Under symmetry, \(\psi_2\) is uncorrelated with the residuals of the second moment condition \(X_i^4-3\). Therefore, the informativeness \(\Delta_2\) depends only on the correlation with the first moment \(X_i\).
Using Definition \ref{def:info},
\[
    \Delta_2 = \frac{(\mathbb{E}[\psi_2 X_i])^2}{\mathbb{E}[\psi_2^2] \mathbb{E}[X_i^2]} = \frac{(\mathbb{E}[\psi_2 X_i])^2}{V_2 \sigma^2}.
\]
The covariance is calculated as
\[
    \mathbb{E}[\psi_2 X_i] = \frac{1}{A_2} (\mathbb{E}[X_i^2] + 4b\mathbb{E}[X_i^4]) = \frac{\sigma^2 + 12b\sigma^4}{1 + 12b\sigma^2} = \sigma^2.
\]
Substituting this back into the expression for \(\Delta_2\),
\[
    \Delta_2 = \frac{(\sigma^2)^2}{V_2 \sigma^2} = \frac{\sigma^2}{V_2} = \frac{V_1}{V_2}.
\]
This shows that informativeness \(\Delta_2\) is exactly the ratio of the efficient variance \(V_1\) to the inflated variance \(V_2\).

\newpage
\section{Appendix: Additional Results}

\subsection{Influence Function for Clustered Data}\label{app:cluster}

We adopt the clustered-sampling framework of \citet*{hansen2019asymptotic}, specialized to the panel setting in which the cluster is the cross-sectional unit. The data consist of \(n\) units indexed \(i = 1, \dots, n\), where unit \(i\) is observed over \(T_i\) periods \(X_i = (X_{i1}, \dots, X_{iT_i})\). As in \citet*{hansen2019asymptotic}, the panel lengths \(T_i\) are nonrandom, units are independent across \(i\), and serial dependence within a unit is left unrestricted.

The GMM estimator is built from the unit-level moment \(g_i(\theta) = \sum_{t=1}^{T_i} g(X_{it}, \theta)\), which aggregates the observation-level moments over the periods of unit \(i\). Writing the unit-level influence function as \(\psi_i = \sum_{t=1}^{T_i}\psi(X_{it})\), the estimator has the asymptotically linear representation
\[
    \sqrt{n}\,(\widehat\theta - \theta_0) = \frac{1}{\sqrt{n}}\sum_{i=1}^n \psi_i + o_p(1),
\]
where the \(\psi_i\) are mean zero and independent across units. As the number of units \(n \to \infty\), the central limit theorem for clustered samples of \citet*{hansen2019asymptotic} gives the asymptotic normality of \(\widehat\theta\), with asymptotic variance estimated by the unit sample analogue of \(\mathrm{Var}(\psi_i)\). The sensitivity and informativeness measures are therefore computed at the unit level, using \(\psi_i\) and the unit moment \(\nu_i = \sum_{t=1}^{T_i}\nu(X_{it})\) in place of their observation-level counterparts.

%\newpage
\subsection{Unique Moments in the BPP Model}
\label{sec:bppmoment}

In the BPP model, the structural equations are
\begin{align*}
    \Delta y_t & = \zeta_t + \epsilon_t + (\theta - 1)\epsilon_{t-1} - \theta \epsilon_{t-2} \\
    \Delta c_t & = \phi \zeta_t + \psi \epsilon_t + \xi_t + \Delta u_t
\end{align*}

We first illustrate the construction of non-collinear moments in the stationary case, where the shock variances are constant over time. Within a three-period window the model implies the following $11$ unique non-collinear second moments.

\begin{enumerate}
    \item \( \mathbb{E}[(\Delta c_t)^2] = \phi^2 \sigma_\zeta^2 + \psi^2 \sigma_\epsilon^2 + \sigma_\xi^2 + 2\sigma_u^2 = C_0 \)

    \item \( \mathbb{E}[\Delta c_t \Delta c_{t-1}] = -\sigma_u^2 = C_1 \)

    \item \( \mathbb{E}[\Delta c_{t-1} \Delta c_{t+1}] = 0 \)

    \item \( \mathbb{E}[(\Delta y_t)^2] = \sigma_\zeta^2 + 2(1 - \theta + \theta^2)\sigma_\epsilon^2 = I_0 \)

    \item \( \mathbb{E}[\Delta y_t \Delta y_{t-1}] = -(1 - \theta)^2 \sigma_\epsilon^2 = I_1 \)

    \item \( \mathbb{E}[\Delta y_{t-1} \Delta y_{t+1}] = -\theta \sigma_\epsilon^2 = I_2 \)

    \item \( \mathbb{E}[\Delta c_t \Delta y_t] = \phi \sigma_\zeta^2 + \psi \sigma_\epsilon^2 = R_0 \)

    \item \( \mathbb{E}[\Delta c_t \Delta y_{t+1}] = \psi (\theta - 1) \sigma_\epsilon^2 = R_1 \)

    \item \( \mathbb{E}[\Delta c_{t-1} \Delta y_{t+1}] = -\psi \theta \sigma_\epsilon^2 = R_2 \)

    \item \( \mathbb{E}[\Delta c_t \Delta y_{t-1}] = 0 \)

    \item \( \mathbb{E}[\Delta c_{t+1} \Delta y_{t-1}] = 0 \)
\end{enumerate}

Let \(\Delta c = (\Delta c_{t-1}, \Delta c_t, \Delta c_{t+1})'\) and \(\Delta y = (\Delta y_{t-1}, \Delta y_t, \Delta y_{t+1})'\). Then the 3-period covariance matrix for each \(t\) is
\[
    \Omega = \begin{pmatrix}
        C_0 & C_1 & 0   & R_0 & R_1 & R_2 \\
        C_1 & C_0 & C_1 & 0   & R_0 & R_1 \\
        0   & C_1 & C_0 & 0   & 0   & R_0 \\
        R_0 & 0   & 0   & I_0 & I_1 & I_2 \\
        R_1 & R_0 & 0   & I_1 & I_0 & I_1 \\
        R_2 & R_1 & R_0 & I_2 & I_1 & I_0
    \end{pmatrix}
\]

The model estimated in Section~\ref{sec:bpp} allows the variances of the permanent shock, the transitory shock, and the consumption innovation to vary over time. The stationary three-period window above is the per-window building block: the full moment vector stacks the analogous non-collinear elements for every three-period window in the panel, and applying the same collinearity removal across all available periods yields $35$ parameters and $325$ non-collinear second-moment conditions.

For reference, we also record the covariance structure in the stationary case without measurement error. The income process consists of a random walk permanent component $\zeta_t$ and a serially uncorrelated transitory component $\epsilon_t$. The growth equation is:
\[
    \Delta y_t = \zeta_t + \Delta v_t = \zeta_t + \epsilon_t - \epsilon_{t-1}
\]

The consumption process follows a random walk with sensitivity to permanent shocks ($\phi$) and transitory shocks ($\psi$), plus an independent innovation $\xi_t$. Measurement error is assumed to be zero in this specific derivation:
\[
    \Delta c_t = \phi \zeta_t + \psi \epsilon_t + \xi_t
\]

Assume that $\zeta_t, \epsilon_t, \xi_t$ are mutually uncorrelated with means zero and variances $\sigma_\zeta^2, \sigma_\epsilon^2, \sigma_\xi^2$, and that all parameters are stationary. Let \(\Delta c = (\Delta c_{t-1}, \Delta c_t, \Delta c_{t+1})'\) and \(\Delta y = (\Delta y_{t-1}, \Delta y_t, \Delta y_{t+1})'\). We obtain the following elements in the 3-period covariance matrix.
\begin{align*}
    \mathbb{E}[(\Delta c_t)^2]                & = \phi^2 \sigma_\zeta^2 + \psi^2 \sigma_\epsilon^2 + \sigma_\xi^2 \equiv C \\
    \mathbb{E}[\Delta c_t \Delta c_{t-1}]     & = 0                                                                        \\
    \mathbb{E}[\Delta c_t \Delta c_{t+1}]     & = 0                                                                        \\
    \mathbb{E}[(\Delta y_t)^2]                & = \sigma_\zeta^2 + 2\sigma_\epsilon^2 \equiv I                             \\
    \mathbb{E}[\Delta y_t \Delta y_{t+1}]     & = -\sigma_\epsilon^2                                                       \\
    \mathbb{E}[\Delta y_{t-1} \Delta y_{t+1}] & = 0                                                                        \\
    \mathbb{E}[\Delta c_t \Delta y_t]         & = \phi \sigma_\zeta^2 + \psi \sigma_\epsilon^2 \equiv R                    \\
    \mathbb{E}[\Delta c_t \Delta y_{t+1}]     & = -\psi \sigma_\epsilon^2                                                  \\
    \mathbb{E}[\Delta c_t \Delta y_{t-1}]     & = 0                                                                        \\
    \mathbb{E}[\Delta c_{t-1} \Delta y_{t+1}] & = 0                                                                        \\
    \mathbb{E}[\Delta c_{t+1} \Delta y_{t-1}] & = 0                                                                        \\
\end{align*}

And the covariance matrix becomes
\[
    \Omega =
    \begin{pmatrix}
        C                      & 0                      & 0 & R                  & -\psi\sigma_\epsilon^2 & 0                      \\
        0                      & C                      & 0 & 0                  & R                      & -\psi\sigma_\epsilon^2 \\
        0                      & 0                      & C & 0                  & 0                      & R                      \\
        R                      & 0                      & 0 & I                  & -\sigma_\epsilon^2     & 0                      \\
        -\psi\sigma_\epsilon^2 & R                      & 0 & -\sigma_\epsilon^2 & I                      & -\sigma_\epsilon^2     \\
        0                      & -\psi\sigma_\epsilon^2 & R & 0                  & -\sigma_\epsilon^2     & I
    \end{pmatrix}
\]

% \newpage
% \section{Appendix: Extra Tables and Figures \label{sec:appTF}}
% \renewcommand{\thetable}{C\arabic{table}}
% \setcounter{table}{0}
% \renewcommand{\thefigure}{C\arabic{figure}}
% \setcounter{figure}{0}

\end{document}